\documentclass[journal, final]{IEEEtran} 
\usepackage{array}
\usepackage{color}
\usepackage{amsfonts}
\usepackage{amssymb}
\usepackage{amsmath}
\usepackage{cite}
\usepackage{balance}
\usepackage[justification=centering]{caption}
\usepackage{subcaption}
\usepackage{textcomp}
\usepackage{gensymb}
\usepackage{float}
\usepackage{listings}
\usepackage{tabularx}
\usepackage{multirow}
\usepackage{amsthm}
\newcommand\underrel[3][]{\mathrel{\mathop{#3}\limits_{%
      \ifx c#1\relax\mathclap{#2}\else#2\fi}}}
\DeclareMathSizes{6}{6}{6}{6}
\usepackage[utf8]{inputenc}
\usepackage[T1]{fontenc}
\usepackage{amsthm}
\usepackage{hyperref}
\newtheorem{theorem}{Theorem}
\newtheorem{corollary}{Corollary}[theorem]
\newtheorem*{proof*}{proof}
\usepackage[english]{babel}
\usepackage{algorithm}
\usepackage[noend]{algpseudocode}
\ifCLASSINFOpdf
  \usepackage[pdftex]{graphicx}
  \graphicspath{{../pdf/}{../jpeg/}}
\else
  \usepackage[dvips]{graphicx}
\fi
\usepackage{lipsum}
\usepackage{cuted}

\usepackage{amsmath}

\begin{document}

\title{Non-Orthogonal Multiple Access for Hybrid VLC-RF Networks with Imperfect Channel State Information}


\author{Ahmed Al Hammadi,~\IEEEmembership{Member,~IEEE,}
        Sami~Muhaidat,~\IEEEmembership{Senior Member,~IEEE,}
        Paschalis~C.~Sofotasios,~\IEEEmembership{Senior Member,~IEEE,}
        Mahmoud~Al-Qutayri,~\IEEEmembership{Senior Member,~IEEE,}
        and~Hany~Elgala,~\IEEEmembership{Senior Member,~IEEE}
\IEEEcompsocitemizethanks{\IEEEcompsocthanksitem
This paper was presented in part in IEEE WCNC 2019, Marrakech, Morocco \cite{WCNC}.  This work was  supported in part by Khalifa University under Grant No. KU/RC1-C2PS-T2/8474000137 and Grant No. KU/FSU-8474000122. 

A. Al Hammadi and M. Al-Qutayri are with the Department of Electrical Engineering and Computer Science, Khalifa University, Abu Dhabi, P.O. Box 127 788, UAE (e-mail: {ahmed.alhammadi; mahmoud.alqutayri}@ku.ac.ae).

S. Muhaidat is with the Center for Cyber-Physical Systems, Department of Electrical Engineering and Computer Science, Khalifa University, P.O. Box 127 788, Abu Dhabi,  UAE (e-mail: muhaidat@ieee.org).

P. C. Sofotasios is with the Center for Cyber-Physical Systems, Department of Electrical Engineering and Computer Science, Khalifa University, P.O. Box 127 788, Abu Dhabi, UAE and also with the Department of Electrical Engineering, Tampere University, FI-33101, Tampere, Finland  (e-mail: p.sofotasios@ieee.org). 

H. Elgala is with the Department of Electrical and Computer Engineering, SUNY at Albany, NY 12222, Albany, USA (e-mail: helgala@albany.edu). 

}
}

\maketitle

\begin{abstract}
The present contribution  proposes a general framework for   the energy efficiency analysis of a hybrid visible light communication (VLC) and Radio Frequency (RF) wireless system, in which both VLC and RF subsystems utilize non-orthogonal multiple access (NOMA) technology. The proposed framework is based on realistic communication scenarios as it takes into account the mobility of users, and assumes imperfect channel-state information (CSI). 
In this context, tractable closed-form expressions are derived for the corresponding  average sum rate of NOMA-VLC and its orthogonal frequency division multiple access (OFDMA)-VLC counterpart. It is shown  extensively that incurred CSI errors have considerable impact on the average energy efficiency of both NOMA-VLC and OFDMA-VLC systems and hence, they should not be neglected in practical designs and deployments. Interestingly, we further demonstrate that the average energy efficiency of the hybrid  NOMA-VLC-RF system outperforms NOMA-VLC system under imperfect CSI. Respective computer simulations corroborate the derived analytic results and interesting theoretical and practical insights are provided, which will be useful in the effective design and deployment of conventional VLC and hybrid VLC-RF  systems. 
\end{abstract}


\section{Introduction}
 
\IEEEPARstart{T}{he} rapidly growing demand for data-intensive applications, such as video streaming, virtual reality (VR), and cloud computing, has led to an explosive growth in the global mobile data traffic resulting in an annual traffic of the order of a zettabyte \cite{1}. This poses challenging requirements for the fifth generation (5G) of wireless networks and beyond, including high spectral efficiency, low latency, and massive connectivity \cite{2} and the references therein. In this context, several emerging technologies have been proposed for ultimately boosting  the data rate in an efficient manner. Examples of such technologies include network densification, massive multiple-input-multiple-output (MIMO) systems, millimeter wave communications, and visible light communications (VLC) \cite{3}.

It has become evident that VLC has recently attracted significant interest as a key technology for supporting the next-generation of wireless networks \cite{3}. As a cost-effective and energy-efficient solution, VLC can potentially achieve considerably high data rates of the order of 100 Gbps \cite{4}. Another critical advatage of VLC is that it occupies the unused frequency spectrum from 400 THz to 800 THz, which is 10,000 times greater than the radio frequency (RF) band. In fact, it has been shown that such complementary approaches increase substantially the performance, efficiency and robustness of wireless systems because they effectively combine  the advantageous features of each single technology \cite{4}. Moreover, a core advantage of VLC is that it is naturally secure due to the nature of light, which cannot penetrate through walls and is confined within its area of illumination, offering a high degree of resource reuse. In addition, VLC is immune to electromagnetic interference (EMI) and thus to the existing RF systems, which is crucial particularly in sensitive areas such as health-care centers, aircraft cabins, and other safety-critical environments \cite{4}.

It is also recalled that multiple access techniques have been also utilized extensively in modern wireless technologies to primarily enhance the spectrum efficiency and address the issue of spectrum crunch. To that end, they have ultimately emerged to be capable of improving several key challenging factors in wireless transmission such as spectral efficiency, robustness, reliability,  energy efficiency and, in general, versatile and stringent quality-of-service (QoS) requirements. In this context, non-orthogonal-multiple-access (NOMA) has attracted significant interests as a breakthrough technology for 5G systems beyond. Since the main requirements of 5G networks revolve around high connectivity, low latency, and ultra-high-speed data rate,   NOMA can be considered   an effective enabling technology to address these requirements. 

The most common variant of NOMA is power-domain NOMA, in which users are multiplexed in the power domain by assigning  unique power levels to different users. 
The main feature of NOMA is that it enables multiple users to simultaneously share the available frequency and time resources, leading to significant performance enhancements in terms of spectral efficiency. This process is referred to as superposition coding (SC), while multiuser detection is realized using successive interference cancellation (SIC) at the receiving terminals.
It is noted here that power-domain based NOMA allocates higher power levels to users that encounter severe fading conditions and less power to the users with more favorable channel conditions. The key motivation underlying  this approach is to allocate particular power levels to different users according to the state of their channels and then perform SIC in order to eliminate the resulting interference and to ultimately achieve a better trade-off between throughput and fairness \cite{8,9,10,11} and the references therein.

Based on the distinct advantages of these technologies, the application of NOMA in VLC systems has been investigated in several reported analyses. Specifically, the authors in \cite{12} analyzed the performance of NOMA based multi-user VLC with gain ratio power allocation strategy. It was shown that the sum rate of VLC could be further enhanced by applying an adaptive tuning to the photo-detectors (PD), field-of-view (FOV), and the semi-angle of the light-emitting diodes (LEDs).  
The superiority of NOMA over OFDMA was shown  in \cite{13} through a performance comparison, taking into account the illumination constraints. Likewise, the performance of NOMA-based VLC with uniformly distributed users was evaluated in \cite{14}, assuming the idealistic case of perfect channel-state-information (CSI).
In \cite{15}, the authors studied the data rate region of indoor VLC networks, whereas the authors in \cite{16} investigated the error-rate performance of   NOMA based VLC systems, with both perfect and imperfect CSI.

\indent More recently, the coexistence between indoor VLC and RF has attracted considerable attention due to its potential to provide enhanced performance in indoor communications  \cite{27}.  In this context,  the main motivations  stems from  the need to overcome the limitations of VLC in duplex transmission scenarios and to ensure ubiquitous  service coverage. To this effect, Shao et al. \cite{18} proposed a hybrid VLC-RF system, where RF is used in the uplink. 
Likewise, Wang et al. \cite{19} studied heterogeneous network, which aims at providing high data rate VLC links alongside  high reliability RF links. The authors in \cite{20} investigated  load balancing in a hybrid VLC-RF system while considering user mobility and associated handover signaling. Also, it was shown \cite{21} that the hybrid VLC-RF system can significantly enhance the overall coverage, which is a key requirement in VLC systems. Likewise, a recent study \cite{22} for hybrid VLC-RF system has shown that this hybrid scheme achieves reduced outage probability and lower power consumption per area.
In another study, the authors in \cite{23} proposed a software-defined heterogeneous hybrid VLC-RF small-cell system. However, a key assumption in the study was the availability of perfect CSI, which is an unrealistic assumption, particularly in highly demanding practical communication  scenarios. 
In \cite{28}, the authors designed an online algorithm to minimize the power consumption of an indoor hybrid VLC-RF network while satisfying the constraint of the illumination level.
In the same context, the authors in  \cite{30}   studied optimal grouping for a hybrid NOMA VLC-RF network, under perfect CSI.
Likewise, \cite{31} addressed the problem of optimal resource allocation in NOMA based hybrid VLC-RF with common backhaul, in order to maximize the achievable data rate.
Finally, the authors in \cite{32} studied the performance of  NOMA based hybrid VLC-RF system in the context of simultaneous wireless and power transfer.

Nevertheless, despite its crucial importance, there have been sporadic results on the achievable energy efficiency levels  of hybrid VLC-RF networks. Recently, the authors in \cite{24} addressed  energy consumption in a multihop VLC-RF network. Likewise, the authors in \cite{25} studied the energy efficiency of an OFDMA based hybrid VLC-RF system through maximizing the system's power efficiency, which is defined as the total system rate per unit power. 
However, in this context and to the best of the authors knowledge, none of the previous contributions investigated the average energy efficiency of NOMA-enabled hybrid VLC-RF systems, assuming uniformly distributed users and imperfect CSI, which both have a detrimental effect on the overall performance of VLC and RF communication systems. 
Based on this, the contributions of the present work are: 
\begin{enumerate}
\item We derive  closed-form expressions for the sum rate for NOMA-VLC and OFDMA-VLC with imperfect CSI systems assuming  the generic case of uniformly distributed users.
\item We derive a closed-form expression for the energy efficiency of a hybrid VLC-RF NOMA system under the realistic assumption of  imperfect CSI.
\item We investigate the performance of both VLC-NOMA and VLC-OFDMA under imperfect CSI and develop  valuable insights into the overall system performance, which could provide interesting guidelines for practical designs.
\item We determine the impact of the CSI error on both NOMA-VLC and NOMA-VLC-RF systems, including its detrimental effect on the achievable  average sum rate and average energy efficiency.
\end{enumerate}
To the best of the authors' knowledge, the offered results have not been reported in the open technical literature. 

The remainder of this paper is organized as follows: Section II describes comprehensively the considered system and channel models. Sections III and IV are  devoted to the derivation of the   average sum rate for the considered NOMA based VLC and hybrid VLC-RF configurations, respectively, under imperfect CSI.  The corresponding energy efficiency of the considered hybrid NOMA VLC-RF system is quantified in Section V, followed by  the respective numerical results and useful discussions  in Section VI. Finally, the paper is concluded with useful remarks in Section VII.

\section{System  and Channel Models}
\begin{figure}
\centering
\includegraphics[height=10cm, width = 8cm]{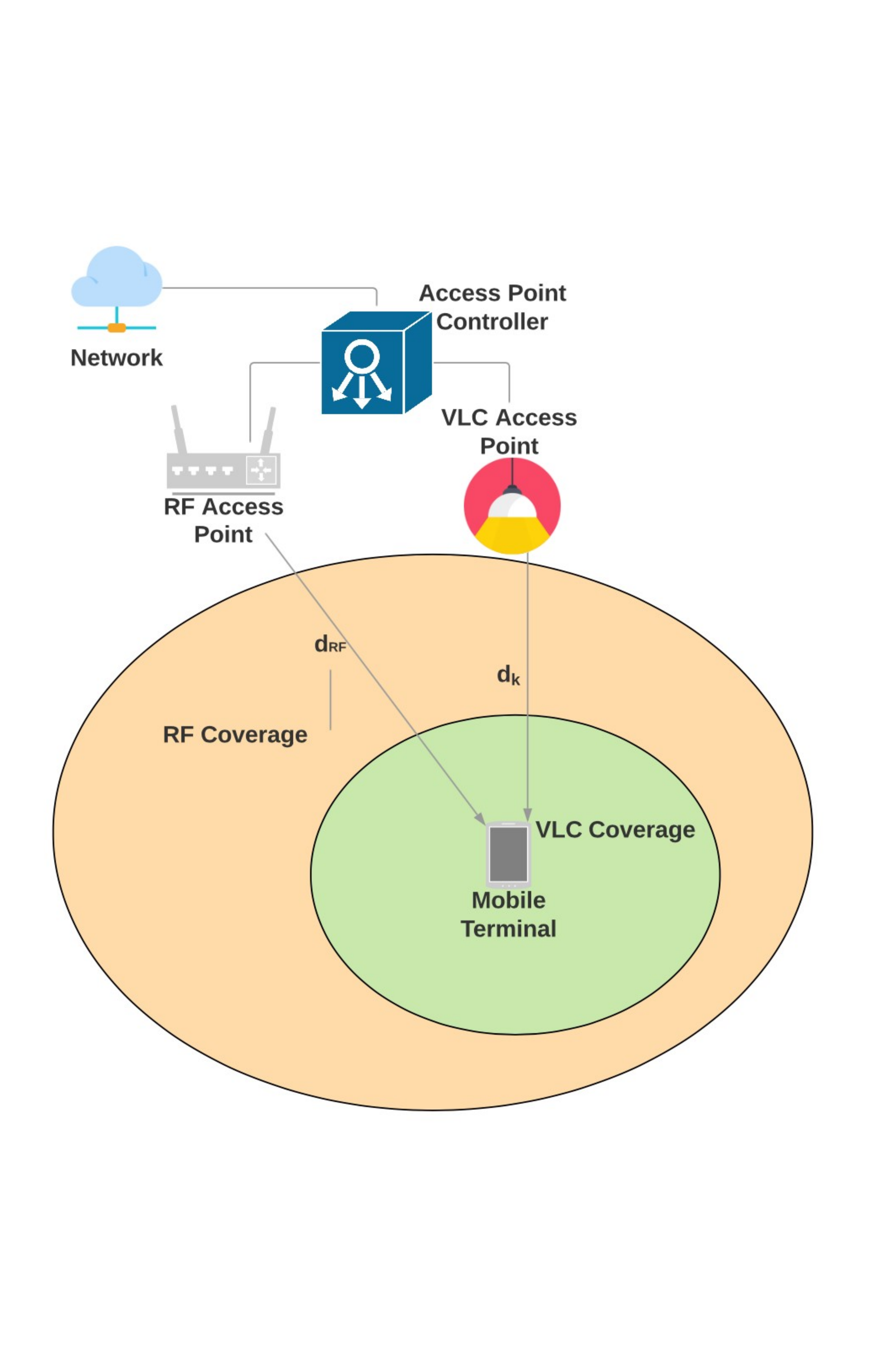}
\caption{The hybrid VLC-RF system model.}
\label{fig2}
\end{figure}
As shown in Fig. 1, we consider a hybrid VLC-RF network operating in an indoor communication  scenario. The network consists of VLC and RF access points (APs) and without loss of generality, the focus is on the downlink. In this context, the directional VLC AP covers a confined area, while the RF AP cover a much wider area, using omni-directional antennas.
 A total number of $K$ users are uniformly distributed, whereas powers denoted as  $Q_{\rm{RF}}$ and $Q_{\rm{VLC}}$ are assumed for the RF and VLC APs, respectively. The sum represents the total power consumed by the hybrid network.
Also, the VLC access point is placed at a height $L$ from the $k_{\rm{th}}$ end user located at the angle ${\theta _i}$ and radius $r_i$ on the polar coordinate plane, whereas the maximum radius for the VLC AP coverage is denoted by $r_e$. 
It is also noted that the RF APs are assigned non-overlapping channels to avoid interference and VLC APs are expected to carry a large portion of the data traffic.
Furthermore, a mobile terminal is equipped with multi-homing capability, where it can aggregate resources from both radio and optical domains.
It is noted here that this analysis adopts NOMA in both networks, where users share the entire RF bandwidth $B_{\rm{RF}}$ and VLC bandwidth $B_{\rm{VLC}}$. 
To this effect, the $P_T$ term denotes the total transmitted power of the VLC-RF network,  namely
\begin{align}
P_{R,i}+P_{V,i}\le  P_T, 
\end{align}
where  $P_{R,i}$ is the allocated power to the $k_{\rm{th}}$ user over the RF link and $P_{V,i}$ is the allocated power for the $k_{\rm{th}}$ user over the VLC link.
It is noted here that although VLC channels incorporate both line-of-sight (LOS) and non-line-of-sight (NLOS) componenets, the energy of the reflected signal in the considered set up is considerably lower than that of the LOS \cite{5}. Based on this, and without loss of generality, this analysis considers only the LOS component of the optical channel gain.

\subsection{VLC Channel Model}
The signal transmitted by the VLC AP can be expressed as 
\begin{align}
x_i^{\rm{VLC}}  = \sum\limits_{i = 1}^K {\alpha _i}\mathop {}\nolimits^{} \sqrt[{}]{{{P_{{\rm{e}}}}}}{s_i} + {I_{{\rm{DC}}}}, 
\end{align}
where $P_{\rm{e}}$ denotes the total electrical power of all the transmitted signals, $I_{\rm DC}$ is the LED DC bias which is essential for intensity modulation based optical baseband transmission, 
$s_i$ represents the modulated symbol of the $k_{\rm{th}}$ user, and $\alpha_i$ is the power allocation coefficient for the corresponding $k_{\rm{th}}$ user. It is assumed that the transmitted signal for each user follows a uniform distribution with zero mean and unit variance.
Based on this and according to the total power constraint of NOMA   systems, the sum of power allocation coefficients must be unity, namely
\begin{align}
\sum\limits_{i = 1}^K {\alpha _i^2 = 1}.
\end{align}
Subsequently, the optical transmission power of the LED can be expressed  as
\begin{align}
{P_{\rm{opt}}} = \eta {\rm E}[x] = \eta {I_{\rm{DC}}}, 
\end{align}
where $\eta$ denotes the efficiency of the LED, which
without loss of generality it is assumed to be unity. 
To this effect, the received signal at the $k_{\rm{th}}$ user can be  expressed as 
\begin{align}
{y_k^{\rm{VLC}}} = \sqrt {{P_{{\rm{e}}}}} {h_k}\left( {\sum\limits_{i = 1}^{k - 1} {{a_i}{s_i} + {a_k}{s_k}}  + \sum\limits_{i = k + 1}^K {{a_i}{s_i}} } \right) + {z_k}, 
\end{align}
where the channel gain $h_k^{{\rm{VLC}}}$ is given by
\begin{equation} \label{5}
h_k^{{\rm{VLC}}} = \frac {(m+1) A R_{\mathrm{ p}}}{2 \pi d_{k}^{2}}\cos ^{m}(\phi _{k}) T(\psi _{k}) g(\psi _{k}) \cos (\psi _{k})
 \end{equation}
 and $z_k$ is the involved  additive white Gaussian noise with zero mean, and variance $\sigma _k^2 = {N_0}B$, with $N_0$ denoting  the noise power spectral density (PSD), $A$ is the area of the photo-detector, $R_p$ denotes the responsivity of the photo-detector, and  $d_{k}$ is the Euclidean distance between the VLC AP and the $k_{\rm{th}}$ user. 
 Also, $U(\psi _{k})$ and $g(\psi _{k})$ denote the optical filter gain and the optical concentrator, respectively.
  Notably, the above equation indicates that the channel gain $h_{k}$ is inversely proportional to the distance of the $k_{\rm{th}}$ user.
 As shown in Fig. 2, the light emitted from the LED follows a Lambertian radiation pattern with order 
\begin{equation}
 p =  -  \frac{1}{\log _2(\cos ({\phi _{1/2}})},
\end{equation}
where ${\phi _{1/2}}$ is the semi-angle of the VLC AP, ${\psi _{\rm FOV}}$ denotes the receiver’s field of view (FOV), whereas ${\psi _k}$ and ${\phi _k}$ denote the angle of incidence and the angle of irradiance, respectively.
 
It is recalled that in power based NOMA systems, users with stronger channel conditions are allocated less signal power, whereas  users with severe channel conditions are allocated more power, which   implies that
${\alpha _1} \ge ... \ge {\alpha _k}... \ge {\alpha _{K - 1}}... \ge {\alpha _K}$.
Without loss of generality,  we assume that the users  in the considered set up are sorted in an ascending order according to their channels, namely
\begin{align}
\left| {h_K^{\rm{VLC}}} \right| \ge \left| {h_{K - 1}^{\rm{VLC}}} \right| \ge ... \ge \left| {h_k^{\rm{VLC}}} \right| \ge ... \ge \left| {h_1^{\rm{VLC}}} \right|.
\end{align}
\begin{figure}
\centering
\includegraphics[height=8cm, width = 7.5cm]{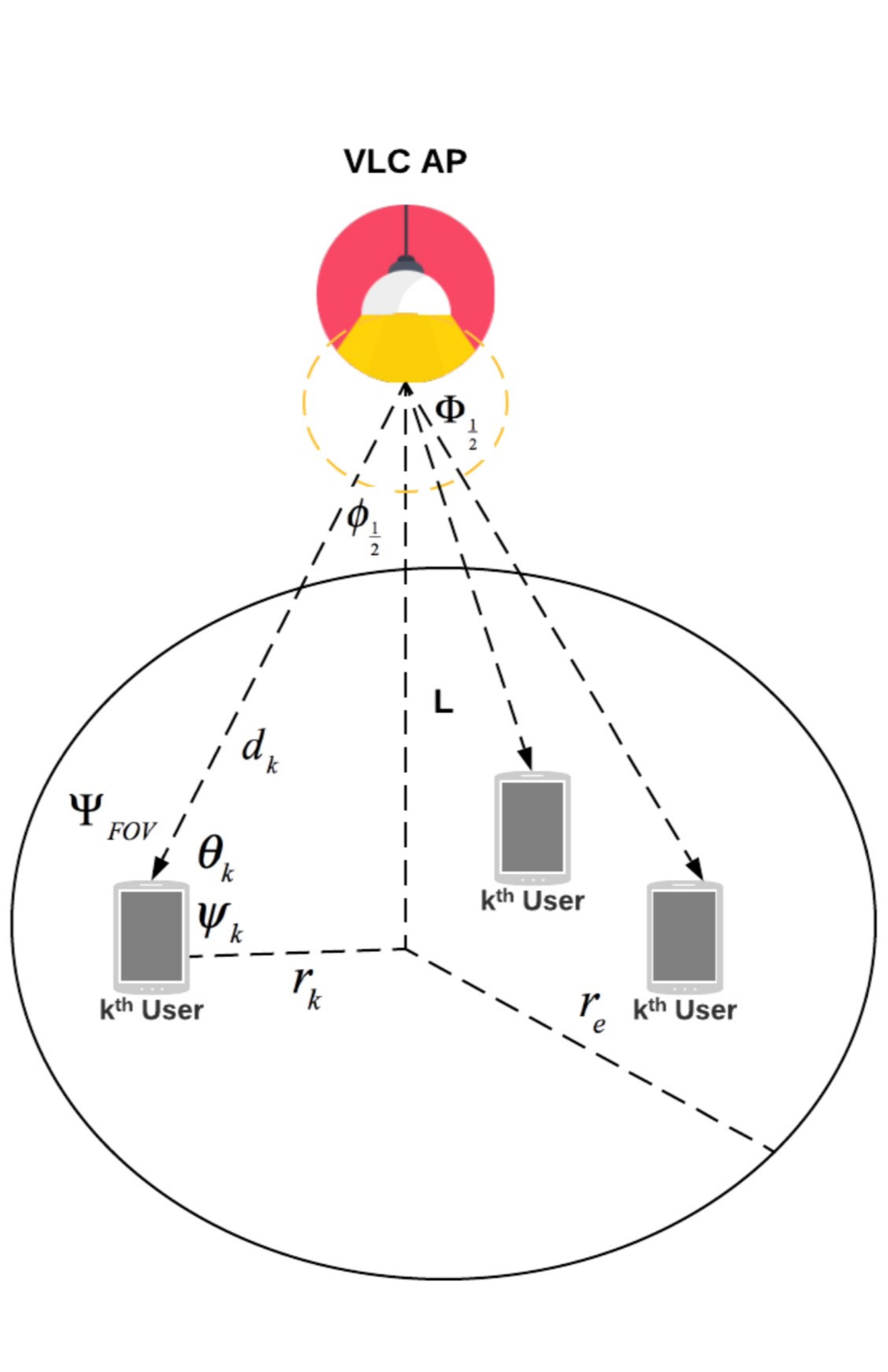}
\caption{VLC Channel Model}
\label{fig2}
\end{figure}
To this effect, in order to decode the signal successively the $k_{\rm{th}}$ user performs SIC in order to remove the signal, s, of the other user(s) with stronger channel condition(s), whereby those signals are treated as noise. To this effect, the achievable data rate per bandwidth is expressed as follows \cite{13}:
\begin{align} \label{7}
 R_{k} = 
 \begin{cases}\log _{2} \left ({ 1+ \frac{ (h_k^{{\rm{VLC}}} \alpha_{k})^{2}}{ \sum \limits _{i=k+1}^{K}(h_k^{{\rm{VLC}}} \alpha_{i})^{2} + \frac {1}{\rho }}}\right ),  & k = 1, \cdots , K-1 \\ 
 $ $ \\
 \log _{2} \left ({ 1+ \rho (h_k^{{\rm{VLC}}} \alpha_{k})^{2}}\right ),  & k = K 
 \end{cases} \notag \\ {}
 \end{align}
where $\rho=P_{e}/{N_0}B$ denotes the signal-to-noise ratio (SNR) at the transmitter. 
It is worth noting that \eqref{7} is conditioned on the requirement that each user  performs the SIC operation successfully.

\subsection{The VLC Channel Model for Uniformly Distributed Users}

The angle of incidence, the angle of irradiance, and the Euclidian distance of the $k_{\rm{th}}$ user in terms of   height $L$ and radical distance $r_k$ are  given by 
\begin{equation}
\cos ({\phi _k}) = \cos ({\psi _k}) =\frac{L}{d_k}
\end{equation}
  and 
  \begin{equation}
  {d_{k}} = \sqrt{r_k^2 + L^2}.
\end{equation}
 Substituting these in \eqref{5}, the DC gain   for the LOS component can be determined, namely 
\begin{equation}
h_k^{{\rm{VLC}}}=\frac {\Xi (m+1)L^{m+1}}{(r_e^{2}+L^{2})^{\frac {m+3}{2}}},  \end{equation}
where 
\begin{equation}
\Xi  = \frac{A{R_p}U({\psi _k})g({\psi _k})}{{2\pi }}
 \end{equation}
is a constant. Furthermore,  because of the uniform distribution of the users, the following probability density function (PDF) is used ${f_{{r_k}}}(r) = 2r/{r_e}$.
Therefore, the PDF of the corresponding channel gain is given by
\begin{align} \label{10}
{f_{{h_k}}}(t) &= \frac{2\left( {\Xi (m + 1){L^{m + 1}}} \right)^{\frac{2}{{m + 3}}}}{{{r_e}^2(p + 3)t^{ \frac{2}{{m + 3}} +1}}}, \, \quad  \,  t \in [{\lambda _{\min }},{\lambda _{\max }}]
\end{align}
where
\begin{align} 
{\lambda _{\min }} = \frac{{{{\Xi^2 (m + 1)^2{L^{2m + 2}}}}}}{{{{({r_e}^2 + {L^2})}^{m + 3}}}}
\end{align}
and 
\begin{align}
 {\lambda _{\max }} =   \frac{\Xi^2 (m + 1)^2 L^{2m + 2}}{L^{2(m + 3)}}. 
\end{align}
Based on this and in order to obtain the corresponding  cumulative distribution function (CDF), we integrate  \eqref{10} over $[{\lambda _{\min }},{\lambda _{\max }}]$, which yields
\begin{align} 
{F_{h_{k}^2}}(t) = 1 + \frac{L^2}{r_e^2} - \frac{(\Xi (m + 1)L^{m + 1})^{\frac{2}{m + 3}}}{r_e^2 t^{\frac{1}{m + 3}} }. 
\end{align}
To this effect and with the aid of  order statistics \cite{book2}, the PDF of the ordered channel gain denoted by ${{f'}_{{h_{k}}}}(t)$ can be obtained as
\begin{align} 
 {{{f'}_{h_{k}^2}}(t)  = \frac{{K! f_{h_k^2}(t)}}{{(k - 1)!(K - k)!}}{F_{h_k^2}}{{(t)}^{k - 1}}{{\left[ {1 - {F_{h_k^2}}(t)} \right]}^{K - k}}{}}, 
\end{align}
which after some algebraic manipulations can be equivalently expressed as follows: 
\begin{align}  \label{new_4}
{{f'}_{h_{k}^2}}(t) =& {\frac{\Omega }{{m + 3}}\frac{{K! {t^{ - \frac{1}{{m + 3}} - 1}}}}{{(k - 1)!(K - k)!}}{{}}}\\
&\nonumber \times \,{{\left( {\frac{\Omega}{{t^{  \frac{1}{{m + 3}}}}} - \frac{{{L^2}}}{{{{r_e}^2}}}} \right)}^{K - k}} \left(1 - \frac{\Omega}{t^{ \frac{1}{m + 3}}} + \frac{L^2}{r_e^2}   \right)^{k - 1}. 
\end{align}
It is evident that \eqref{new_4} has a simple algebraic representation since it is expressed in terms of elementary functions. 
As a result, it is convenient to handle both analytically and numerically.

\subsection{The RF Channel Model} 

Following the same principles, the signal transmitted  by the RF AP can be expressed as 
\begin{align}
x_i^{\rm{RF}} = \sum\limits_{i = 1}^K {\sqrt {P_i^{\rm{RF}}} } S_i^{\rm{RF}}, 
\end{align}
where $S_i^{\rm{RF}}$ is the modulated symbol, which represents the transmitted symbol of the $k_{\rm{th}}$ user.
We also let ${h_K^{\rm{RF}}}$ denote the channel gain from the RF AP to the $k_{\rm{th}}$ user. 
Based on this and without loss of generality, it is assumed that the values of ${h_k^{\rm{RF}}}$ of the $K$ users are perfectly known and are sorted as 
\begin{align}
\left| {h_K^{\rm{RF}}} \right| \ge \left| {h_{K - 1}^{\rm{RF}}} \right| \ge ... \ge \left| {h_k^{\rm{RF}}} \right| \ge ... \ge \left| {h_1^{\rm{RF}}} \right|.
\end{align}
By applying the NOMA principle, the signal received by the $k_{\rm{th}}$ user in the RF channel can be represented  as
\begin{align}  
y_k^{\rm{RF}} = h_k^{\rm{RF}}\sum\limits_{i = 1}^K {\sqrt {{\alpha _{i,r}}P} {s_{i,r}} + {z_{i,r}}},  
\end{align}
where $ h_k^{\rm{RF}}$ stands for  the RF channel gain, namely 
\begin{align}
 h_k^{\rm{RF}} = \frac{{{H_{k}}}}{{d_k^{PL/2}}},
\end{align}
where ${H_{k}} \sim \mathcal{CN} (0,1)$ with  $d_k$ denoting the distance between the RF AP and the $k_{\rm{th}}$ user, and {$\rm{PL}$} denoting the RF path loss exponent.
It is noted here that although the system can achieve its optimal performance when the CSI channel is considered perfect, obtaining a perfect or near-perfect CSI is not technically practical in realistic communication scenarios. Subsequently, it is crucial to obtain an insight into the system performance under realistic conditions in which the CSI channel can be typically imperfect.

\subsection{Imperfect CSI Model} 
 
Unlike the majority of the previous related contributions which assumed perfect CSI knowledge, this work assumes the practical case of imperfect CSI for the underlying RF-VLC system model.
Typically, CSI can be determined at the receiver with the aid of pilots symbols. The quantized channel coefficients are sent to the transmitters over an RF or infrared (IR) uplink. Accordingly, the uncertainty in the VLC channel estimation arises from the noise in the downlink and uplink channels as well as from the mobility of users in indoor environments. Moreover, it is noted that the  analog-to-digital and digital-to-analog (AD/DA) conversion of the channel estimates introduces quantization errors that also contribute to the incurred channel uncertainty, which is ultimately detrimental to the overall system performance.

\subsubsection{VLC Imperfect CSI Model}

By assuming the minimum mean squared error channel estimation model,  the channel coefficient for the VLC link can be represented  as
\begin{align}  
h_k^{{\rm{VLC}}} = \hat h_k^{{\rm{VLC}}} + e_k^{\rm{VLC}}, 
\end{align}
where ${{\hat h}_{k}}$  $\sim \mathcal{N}(0,1 - \sigma _e^2)$ denotes the estimated channel gain and $e_{k}$ is the estimated error in the channel which follows a Gaussian distribution with zero mean and variance $\sigma _e^2$. It is worth nothing that the random variables ${{\hat h}_{k}}$ and $e_{k}$ are uncorrelated and the maximum achievable data rate for the VLC channel according to Shannon formula is given by
\begin{align}    \label{20}
{C_{\rm{VLC}}} =B_{\rm{VLC}}{\log _2}(1 + \rm{SINR}), 
\end{align}
where $B_{\rm{VLC}}$ is the bandwidth of the channel. 
Hence,  the signal-to-interference-plus-noise ratio (SINR) of the $k_{\rm{th}}$ user over the VLC channel under imperfect CSI is given by 
\begin{align} \label{21}
{{\rm{SINR}}_{k}^{{\rm{VLC}}}} =
 \begin{cases}
\frac{\hat{h}_k^2 \alpha_k^2}{\sum\limits_{i = k + 1}^K \hat{h}_k^2 \alpha_i^2 + \frac{1}{\rho} + \sigma_e^2 },  & \quad k = 1, \cdots , K-1 \\
$ $ \\
  1 + \rho \hat{h}_k^2 \alpha_k^2 + \sigma_e^2,  & \quad k = K \end{cases} 
 \notag \\ {}
 \end{align}
Based on this and substituting \eqref{21} into \eqref{20},  the maximum achievable data rate for the $k_{\rm{th}}$ user under imperfect CSI is given by
\begin{align} 
 R_k^{{\rm{VLC}}} =
 \begin{cases}
\log_2 \left(1 + \frac{\hat{h}_k^2 \alpha_k^2}{\sum\limits_{i = k + 1}^K \hat{h}_k^2 \alpha_i^2 + \frac{1}{\rho} + \sigma_e^2 }\right),  \, \,    k = 1, \cdots , K-1 \\
$ $ \\
\log_2 \left(  1 + \rho \hat{h}_k^2 \alpha_k^2 + \sigma_e^2\right),  \qquad \quad  k = K \end{cases} 
 \notag \\ {}
 \end{align}
which also has a simple algebraic representation. 

\subsubsection{RF Imperfect CSI Model} 

Based on the same approach i.e. using the minimum mean squared error (MMSE) channel estimation model, the channel coefficient for the RF link can be similarly modeled as
\begin{align}  
h_k^{{\rm{RF}}} = \hat h_k^{{\rm{RF}}} + e_k^{\rm{RF}}, 
\end{align}
where $ h_k^{\rm{RF}}$  $\sim \mathcal{CN}(0,1 - \sigma _e^2)$ denotes the estimated channel gain and $e_k^{\rm{RF}}$ is the estimated error in the channel, which follows a complex Gaussian distribution with   zero mean and variance 
$\sigma_e^2$. Also,  it is assumed again that the random variables $h_k^{\rm{RF}}$ and $e_k^{\rm{RF}}$ are  uncorrelated. 
To this effect, the maximum achievable data rate for the  channel according to Shannon formula is readily  expressed as 
\begin{align}   \label{25}
{C_{\rm{RF}}} = B_{\rm{RF}}{\log _2}(1 + \rm{SINR}), 
\end{align}
where $B_{\rm{RF}}$ is the bandwidth of the RF channel. 
Also, the SINR for the $k_{\rm{th}}$ user over an RF channel under imperfect CSI is expressed as 
\begin{align} \label{26}
{{\rm{SINR}}_{k}^{{\rm{RF}}}} =
 \begin{cases}
\frac{ \alpha_k (\hat{h}_k^{\rm RF})^2}{\sum\limits_{i = k + 1}^K \alpha_i  (\hat{h}_k^{\rm RF})^2 + \frac{1}{\rho} + \sigma_e^2 },  & \quad k = 1, \cdots , K-1 \\
$ $ \\
  1 + \rho \alpha_k (\hat{h}_k^{\rm RF})^2  + \sigma_e^2,   & \quad  k = K \end{cases} 
 \notag \\ {}
 \end{align}

To this effect, substituting \eqref{26} into \eqref{25}, the maximum achievable data rate for the $k_{\rm{th}}$ user under imperfect CSI is given by
\begin{align}  
{{{R}_{k}^{\rm{RF}}}} = 
 \begin{cases}
\log_2 \left( 1 + \frac{ \alpha_k (\hat{h}_k^{\rm RF})^2}{\sum\limits_{i = k + 1}^K \alpha_i  (\hat{h}_k^{\rm RF})^2 + \frac{1}{\rho} + \sigma_e^2 }\right),  \,     k = 1, \cdots , K-1 \\
$ $ \\
 \log_2 \left( 1 + \rho \alpha_k (\hat{h}_k^{\rm RF})^2  + \sigma_e^2\right),  \,    \qquad  k = K \end{cases} 
 \notag \\ {}
 \end{align}
which has  similarly a  simple algebraic representation as the maximum achievable data rate in the VLC channel.

\section{Average Sum Rate of VLC under imperfect CSI}

In this section, a closed-form sum rate expression for NOMA-based VLC systems  is derived. Additionally,  the sum rate for the corresponding OFDMA-based VLC system is derived for both perfect and imperfect CSI for the sake of benchmarking as it is subsequently compared    with the results of the considered set up.

\subsection{Average Sum Rate of Uniformly Distributed users}
In this subsection, the average sum rate  is derived of both NOMA and OFDMA VLC systems under imperfect CSI, assuming uniformly distributed users.

\subsubsection{NOMA Sum Rate with Imperfect CSI}

The NOMA based sum rate with imperfect CSI is introduced in the following theorem. 

\begin{theorem}
For a $K$ number of uniformly distributed users and an arbitrary power allocation strategy, the average sum rate of NOMA-VLC under imperfect CSI is expressed by the exact closed-form expression in \eqref{27}, at the top of the page, 
\begin{figure*}
\begin{equation} \label{27}
\begin{split}
& \, {{{\hat R}_{{\rm{VLC}}}}^{{\rm{NOMA}}} = \frac{{\Xi K}}{{\ln (2)(m + 3)}}\left\lbrace {\sum\limits_{l = 0}^{K - 1} { \frac{{(K - 1)!{{( - \Xi )}^l}}}{{l!\,(K - 1 - l)!({v_1} + 1)}}{{\left( {\frac{{{L^2}}}{{r_{\rm{e}}^2}} + 1} \right)}^{K - 1 - l}}\left[ {\Omega ({\lambda _{\max }},{v_1},{b_1}) - \Omega ({\lambda _{\min }},{v_1},{b_1})} \right] } } \right.}\\
& {   \qquad \quad \quad  + \sum\limits_{k = 1}^{K - 1} {\sum\limits_{p = 0}^{k - 1} {\sum\limits_{q = 0}^{K - k} { \frac{{(K)!{{(\Xi )}^{p + q}}{{( - 1)}^{p + K - k - q}}}}{{p!\,(k - 1 - p)q!(K - k - q)!}}{{\left( {\frac{{{L^2}}}{{r_{\rm{e}}^2}} + 1} \right)}^{k - 1 - p}}\left. {{{\left( {\frac{{{L^2}}}{{r_{\rm{e}}^2}}} \right)}^{K - k - q}}\left[ {\Omega ({\lambda _{\max }},{v_2},{b_2}) - \Omega ({\lambda _{\min }},{v_2},{b_2})} \right]} \right\rbrace} } } }. 
\end{split}
\end{equation}
\hrulefill
\end{figure*}
where
\begin{equation} \label{new_1}
\begin{split}
\Omega (x,y,z) =& \left( {x(y + 1) + \frac{{{1}}}{{{{ z }^{y + 1}}}}} \right)  \ln (1 +  zx)\\
&   + \sum\limits_{i = 1}^{y + 1} {\left( {\frac{{-{x_{^{y - i + 2}}}{{(1)}^{i - 1}}}}{{\left( {y - i + 2} \right){{z}^{i - 1}}}}} \right)}.
\end{split}
\end{equation}
\end{theorem}
\begin{proof}The average sum rate of NOMA VLC under imperfect CSI is given by 
\begin{align}   \label{new_2}
&{{\hat R}_{{\rm{VLC}}}}^{{\rm{NOMA}}} = \sum\limits_{k = 1}^M E \left[ {{R_k}} \right]\\
& = \sum\limits_{k = 1}^{K - 1} {\underbrace {\int_{{\lambda _{{\rm{min}}}}}^{{\lambda _{{\rm{max}}}}} {{{\log }_2}\left( {1 + \frac{{{{({\alpha _k})}^2}t}}{{\sum\limits_{i = k + 1}^K {{{({\alpha _i})}^2}} t + \frac{1}{\rho } + \sigma _e^2}}} \right){{f'}_{h_k^2}}\left( t \right)\rm{dt}}}_{{Q_s}}}  \nonumber \\
 & \qquad   + \,\underbrace {\int_{{\lambda _{{\rm{min}}}}}^{{\lambda _{{\rm{max}}}}} {{{\log }_2}\left( {1 + \rho \alpha _K^2t + \sigma _e^2} \right)} {{f'}_{h_K^2}}\left( t \right)\rm{dt}}_{{V_s}} \label{30}
\end{align}
where $Q_s$ denotes the ergodic data rate for the $k_{\rm{th}}$ user, $k \in \{ 1, \cdots ,K - 1\}$, and $V_s$ denotes the ergodic data rate for the $k_{\rm{th}}$ user. 
Therefore, it is evident that the derivation of a closed-form expression for \eqref{new_2} is subject to analytic solutions of the involved two integrals. To that end, applying the binomial expansion in the above, $V_s$ can be written as:
\begin{equation}\label{31}
\begin{split}   
{V_s} = \frac{{\Xi K}}{{2(m + 3)}}\sum\limits_{l = 0}^{K - 1}   \frac{{(K - 1)!{{( - \Xi )}^l}}}{{l!\,(K - 1 - l)!}}{\left( {\frac{{{L^2}}}{{r_{\rm{e}}^2}} + 1} \right)^{K - 1 - l}}\\
 \times \,\underbrace {\int_{{\lambda _{{\rm{min}}}}}^{{\lambda _{{\rm{max}}}}} {{{\log }_2}} \left[ {1 + \left( {\rho {{({}{\alpha _k})}^2}t + \sigma _e^2} \right)} \right]{t^{ - \frac{{l + 1}}{{m + 3}} - 1}}{\rm{dt}}}_{Vsi}. 
\end{split}
\end{equation}
Evidently, we have to solve the integral $V_{si}$ in order to deduce the final solution.
To this end, using \cite[2.729.l]{book2}, and by defining
\begin{equation}
 b_1 =  {\rho {{({\alpha _k})}^2} + \sigma _e^2}  
\end{equation}
and
\begin{equation}
v_1 =  - \frac{{l + 1}}{{m + 3}} - 1
\end{equation}
$V_{\rm{si}}$ can be readily obtained as follows:
\begin{equation} \label{32}
\begin{split}  
&{{V_{si}} = \frac{1}{{\ln (2)({v_1} + 1)}} \left( {{t^{{v_1} + 1}} - \frac{{{{( - a)}^{{v_1} + 1}}}}{{{{\left( {\rho {{({\alpha _k})}^2} + \sigma _e^2} \right)}^{{v_1} + 1}}}}} \right)}\\&
{ \qquad  \times \ln \left( {1 + \left( {\rho {{({\alpha _k})}^2} + \sigma _e^2} \right)t} \right)}\\ &
{ + \frac{1}{{\ln (2)({v_1} + 1)}}\sum\limits_{i = 1}^{{v_1} + 1} {\left( {\frac{{{{( - 1)}^i}{t^{{v_1} - i + 2}}{a^{i - 1}}}}{{\left( {{v_1} - i + 2} \right){{\left( {\rho {{({\alpha _k})}^2} + \sigma _e^2} \right)}^{i - 1}}}}} \right)}.}
\end{split}
\end{equation}
Based on this and substituting \eqref{32} into \eqref{31}, the following expression is obtained
\begin{equation} \label{33}
\begin{split}   
{V_s} =   \Xi K  \sum\limits_{l = 0}^{K - 1}  & {\left\{  {\frac{{(-1)^l(K - 1)!{{ \Xi }^l}}}{{l!\,(K - 1 - l)!({v_1} + 1)}}{{\left( {\frac{{{L^2}}}{{r_{\rm{e}}^2}} + 1} \right)}^{K - 1 - l}}} \right.} \\
& \left. \,   { \times \frac{  {\Omega ({\lambda _{\max }},{b_1},{v_1}) - \Omega ({\lambda _{\min }},{b_1},{v_1})}  }{\ln (2)(m + 3)}} \right\},
\end{split}
\end{equation}
where
\begin{align}  
\nonumber \Omega (x,y,z) = \left( {{x_{^{y + 1}}} - \frac{{{{( - 1)}^{y + 1}}}}{{{{\left( z \right)}^{y + 1}}}}} \right) \ln (1 + \left( z \right)x) \\
+ \sum\limits_{i = 1}^{{v_1} + 1} {\left( {\frac{{{{( - 1)}^i}{x_{^{{v_1} - i + 2}}}{{(1)}^{i - 1}}}}{{\left( {y - i + 2} \right){{\left( z \right)}^{i - 1}}}}} \right)}.
\end{align}
Similarly, by applying binomial expansion for the integral $Q_s$ in \eqref{30}, we obtain
\begin{align} 
\nonumber {Q_s} =\frac{\Xi }{{(m + 3)}}\sum\limits_{k = 1}^{K - 1} {\sum\limits_{p = 0}^{k - 1} {\sum\limits_{q = 0}^{K - k} }}{\frac{{K!{{\Xi }^{p + q}}{{( - 1)}^{p + K - k - q}}}}{{p!\,(k - 1 - p)q!(K - k - q)!}}}
\end{align}
\begin{align} 
\times {{\left( {\frac{{{L^2}}}{{r_{\rm{e}}^2}} + 1} \right)}^{k - 1 - p}}{{\left( {\frac{{{L^2}}}{{r_{\rm{e}}^2}}} \right)}^{K - k - q}}{Q_{si}},
\end{align}
where 
\begin{equation}
{Q_{si}} = \int_{{\lambda _{{\rm{min}}}}}^{{\lambda _{{\rm{max}}}}} {{{\log }_2}\left( {1 + \frac{{{{({\alpha _k})}^2}t}}{{\sum\limits_{i = k + 1}^K {{{({\alpha _i})}^2}} t + \frac{1}{\rho } + \sigma _e^2}}} \right)}{\rm{dt}}. 
\end{equation}
Next, $Q_{\rm{si}}$ can be deduced using the same previous steps. Hence, by first defining the following parameters:
\begin{align}
\large
b_2 = \frac{{\rho ({{({\alpha _k})}^2} - \sigma _e^2)}}{{\sum\limits_{i = k + 1}^K {{{{\alpha _i^2}}}} }}
\end{align}
and
\begin{align}
v_2 =  - \frac{{p + q + 1}}{{m + 3}} - 1
\end{align}
and carrying out  some algebraic manipulations yields \eqref{36}, at the top of the next page. 
\begin{figure*}
\begin{equation}   \label{36}
\begin{split}
 {Q_s} =   \sum\limits_{k = 1}^{K - 1} {\sum\limits_{p = 0}^{k - 1} {\sum\limits_{q = 0}^{K - k} { \frac{{K!{{\Xi}^{p + q+1}}{{( - 1)}^{p + K - k - q}(L^2 + r_e^2)^{ k - p - 1} L^{2(K - k - q)}}}}{{ \ln (2)(m + 3) p!\,(k - 1 - p)q!(K - k - q)! r_e^{2(K - q - p - 1)}}}{{ }{}}} } } 
 { {{}{}} \times \{ {\Omega ({\lambda _{\max }},b2,v2) - \Omega ({\lambda _{\min }},b2,v2)} \} }.
\end{split}
\end{equation}
\hrulefill
\end{figure*}
Hence, by substituting \eqref{36} and \eqref{33} into \eqref{30} yields  \eqref{27}, which completes the proof.
\end{proof}
It is noteworthy that the results reported in \cite{13} include complex and constrained  special functions. In addition, they are limited to the simplistic assumption of perfect CSI knowledge as they do not take into account any incurred CSI errors. 
Hence, the consideration of  CSI errors in Theorem 1 is more practical and therefore more useful when considering realistic communication scenarios in emerging technologies that are typically characterized by demanding requirements and stringent quality of service targets. 
Another advantage of the offered result in Theorem 1 is its simple algebraic form since it does not include special functions. This renders it a versatile result since it is convenient to handle both analytically and numerically.

\subsubsection{OFDMA Sum Rate with Imperfect CSI}

In what follows we  quantify  the  sum rate for the  corresponding OFDMA counterpart.

\begin{theorem}
For $K$ number of uniformly distributed users, and an arbitrary power allocation strategy, the average sum rate of OFDMA VLC under imperfect CSI can be expressed by the closed-form  representation in
\begin{figure*}
\begin{equation}\label{40}
 {{\hat R}_{{\rm{VLC}}}}^{{\rm{OFDMA}}} =  \sum\limits_{k = 1}^K {\sum\limits_{p = 0}^{k - 1} {\sum\limits_{q = 0}^{K - k} { \frac{{K!{{\Xi}^{p + q + 1}}{{( - 1)}^{p + K - k - q} \left[ {\Omega ({\lambda _{\max }},{v_3},{b_3}) - \Omega ({\lambda _{\min }},{v_3},{b_3})} \right] }}}{{2\ln (2) p!\,(k - 1 - p)q!(K - k - q)!({v_3} + 1) (m + 3) }}{{\left( {\frac{{{L^2}}}{{r_{\rm{e}}^2}}} \right)}^{K - k - q}}} } } {{\left( {\frac{{{L^2}}}{{r_{\rm{e}}^2}} + 1} \right)}^{k - 1 - p}} 
\end{equation}
\hrulefill
\end{figure*}
 \eqref{40}, at the top of the next page.
\end{theorem}
\begin{proof}
The average sum rate of OFDMA VLC under imperfect CSI is obtained by the following representation:
\begin{align}   
 {{\hat R_{\rm{VLC}}}^{\rm{OFDMA}}} &= \sum\limits_{k = 1}^K  \int_{\lambda_{\rm min}}^{\lambda_{\rm max}}  \log_2 \left(1 + \frac{\rho w_k t}{b_{nk}(1 + \rho \sigma_e^2)} \right) \frac{ w_k f'_{h_k^2}(t)}{2} \rm{dt}
\end{align}
where $w_k$ is the fraction of bandwidth occupied by the $k_{\rm{th}}$ user, and $b_{n_{k}}$ is the fraction of the power allocated to the $k_{\rm{th}}$ user.
By applying the binomial expansion, the average sum rate of OFDMA VLC under imperfect CSI is given by 
\begin{figure*}
\begin{flalign}    \label{39}
\nonumber {\hat R_{\rm{VLC}}}^{\rm{OFDMA}}  &=  \frac{ \Xi }{{2(m + 3)}}{\sum\limits_{k = 1}^K {\sum\limits_{p = 0}^{k - 1} \sum\limits_{q = 0}^{K - k} { \frac{{K!{{\Xi}^{p + q}}{{( - 1)}^{p + K - k - q}}}}{{p!\,(k - 1 - p)q!(K - k - q)!}} }}}
\nonumber    {{\left( {\frac{{{L^2}}}{{r_{\rm{e}}^2}} + 1} \right)}^{k - 1 - p}} \left( {\frac{{{L^2}}}{{r_{\rm{e}}^2}}} \right)^{K - k - q}\\
&  \qquad {\times \sum\limits_{k = 1}^K {\underbrace {\int_{{\lambda _{{\rm{min}}}}}^{{\lambda _{{\rm{max}}}}} {  {\frac{1}{2}{w_k}{{\log }_2}\left( {1 + \frac{w_k \rho t}{b_{nk} \sigma_e^2} } \right)
{{f'}_{h_k^2}}(t){\rm{d}}t} } }_{{T_s}}} }&&
\end{flalign}
\hrulefill
\end{figure*}
\eqref{39} at the top of the next page, where $T_{s}$ can be expressed as 
\begin{align}  \label{38}
T_{s} = \frac{{\Omega ({\lambda _{\max }},{v_3},{b_3}) - \Omega ({\lambda _{\min }},{v_3},{b_3})}}{{2\ln (2)}}. 
\end{align}
Subsequently, using \cite[2.729.l]{30}, and by defining:
\begin{align}
b_3 = \frac{\rho w_k}{b_{nk} (1 + \rho \sigma_e^2)}
\end{align}
and
\begin{align}
v_3 =  - \frac{{p + q + 1}}{{(m + 3)}} - 1, 
\end{align}
it follows that the exact closed-form expression in \eqref{40} is obtained by substituting \eqref{38} into \eqref{39}, which completes the proof of the theorem.
\end{proof}

\begin{figure*}
\begin{equation} \label{new_3}
\begin{split}
{{\hat R}_{\rm{RF}}}^{\rm{NOMA}} \approx  & 
\sum_{k = 2}^{K} \sum_{r = 0}^{k - 1}      \binom{K}{k} \binom{k - 1}{r}  \frac{k (-1)^r}{\ln(2) (r + k)} \sum_{j = 1}^{r + K - k + 1}  \binom{r + K - k + 1}{j} \frac{(-1)^j \pi^j}{n^j D^j} \phi(\beta_{R, k-1})\\
& - \sum_{k = 1}^{K} \sum_{r = 0}^{k - 1}      \binom{K}{k} \binom{k - 1}{r}  \frac{k (-1)^r}{\ln(2) (r + k)} \sum_{j = 1}^{r + K - k + 1}  \binom{r + K - k + 1}{j} \frac{(-1)^j \pi^j}{n^j D^j} \phi(\beta_{R, k})
\end{split}
\end{equation}
\hrulefill
\end{figure*}

\section{Average Sum Rate of NOMA-RF under imperfect CSI}

The average sum rate of a NOMA based RF system with imperfect CSI can be obtained using \eqref{new_3} at the top of the page, where \cite{26}
\begin{equation} 
\phi (x) = \sum\limits_{{t_1} +  \cdots  + {t_n} = j} {\frac{{j!}}{{{t_1}!{t_2}! \cdots {t_n}!}}} \prod\limits_{i = 1}^n {{{\left( {\left| {\sin \frac{{2i - 1}}{{2n}}\pi } \right|{x_i}} \right)}^{{t_i}}}} h(x) 
\end{equation} 
with $j$ denoting  the summation of all sequences of non-negative integer indices from $t_1$ until $t_n$, and 
\begin{equation}
\begin{split}
h(x) =& \exp \left( {\frac{{(\sigma _z^2 + {P_{R,k}}\sigma _\varepsilon ^2)}}{{{P_{R,k}}z}}\sum\limits_{i = 1}^n {\frac{{{t_i}}}{{x_i^{ - {\beta _R}} - \sigma _\varepsilon ^2}}} } \right)\\
& \times {E_1}\left( {\frac{{(\sigma _z^2 + {P_{R,k}}\sigma _\varepsilon ^2)}}{{{P_{R,k}}z}}\sum\limits_{i = 1}^n {\frac{{{t_i}}}{{x_i^{ - {\beta _R}} - \sigma _\varepsilon ^2}}} } \right)
\end{split}
\end{equation}
where $E_1(\cdot)$ denotes the exponential integral \cite{Tables}.  This function can be readily computed because it is a standard built in function in popular scientific software packages such as MAPLE, MATLAB and MATHEMATICA.  

To the best of the authors knowledge, the offered results have not been reported in the open literature. 
Capitalizing on the results, the  performance of the considered set up in terms of the average sum rate of NOMA-RF under imperfect CSI is quantified in detail in th enumerical results section, where useful theoretical and practical insights are developed. 

\section{Energy Efficiency of Hybrid NOMA VLC-RF System}

In this section, the total achievable data rate and the energy efficiency of the NOMA-VLC-RF network are derived assuming imperfect knowledge of the CSI.
Following \cite{25}, we define the total average sum rate in a hybrid VLC-RF network as follows:
\begin{equation} \label{44}
{R_{\rm{SUM}}} = {B_{\rm{RF}}}\left( {{\beta _{\rm{RF}}}\hat R_{\rm{RF}}^{\rm{NOMA}}} \right) + {B_{\rm{VLC}}}\left( {{\beta _{\rm{VLC}}}\hat R_{\rm{VLC}}^{\rm{NOMA}}} \right), 
\end{equation}
where ${\beta_{\rm{RF}}}$ is the probability of available LOS (dominant) component  for the RF link, and ${\beta_{\rm{VLC}}}$ is the corresponding  LOS availability probability for the VLC link.
It is noted that the above representation is useful in determining the   average energy efficiency of the considered hybrid set up.
\begin{corollary}
The average energy efficiency of the considered  hybrid VLC-RF network can be determined as follows:
\begin{align} \label{45}
\hat \xi  = {\rm{ }}\frac{{{B_{\rm{RF}}}\left( {{\beta _{\rm{RF}}}\hat R_{\rm{RF}}^{\rm{NOMA}}} \right) + {B_{\rm{VLC}}}\left( {{\beta _{\rm{VLC}}}\hat R_{\rm{VLC}}^{\rm{NOMA}}} \right)}}{{{Q_{\rm{VLC}}} + {Q_{\rm{RF}}} + \sum\limits_{i = 1}^N {{P_{\rm{RF},i}}} }}.
\end{align}
\end{corollary}
\begin{proof}
The total achievable data rate of the considered scheme  is expressed in bits/s, whereas the energy efficiency is expressed  in bit/J. Accordingly, the energy efficiency of the underlying scheme is given by:
\begin{align} \label{46}
\xi  = \frac{{R_{\text{{\sc{SUM}}}}}}{{{Q_{\rm{VLC}}} + {Q_{\rm{RF}}} + \sum\limits_{i = 1}^N {{P_{\rm{RF},i}}} }}, 
\end{align}
where $R_{\rm{Sum}}$ denotes the total achievable sum rate of NOMA-VLC-RF and  $\sum\nolimits_{i = 1}^N {{P_{\rm{RF},i}}} $ is the total transmission power of the RF network for $N$ users. Based on this and by substituting  \eqref{44} into \eqref{46},  equation  \eqref{45} is deduced, which completes the proof.
\end{proof}

\section{Results and Discussion}
\indent 
\begin{table}[]
\centering
\caption{Simulation parameters}
\label{my-label}
\begin{tabular}{|l|l|}
\hline
\textbf{Parameter}                     & \textbf{Value}                     \\ \hline
\small Vertical separation between the LED and PDs,  $L$ & \small 2.15 m                             \\ \hline
\small Cell radius, $r_e$                               & \small 3.6 m                              \\ \hline
\small LED semi-angle                                   & \small 45$^\circ$                                \\ \hline
\small Total signal power, $P_{e}$                  &\small  0.25 W                             \\ \hline
\small PD FOV, ${\Psi _{fov}}$                     &\small  60$^\circ$                                 \\ \hline
\small PD responsivity $R_p$                            & \small 0.4 A/W                            \\ \hline
\small PD detection area, $A$                           &\small  1 ${\rm{cm}}^{2}$                  \\ \hline
\small Reflective index, $n$                            &\small  1.5                                \\ \hline
\small Optical filter gain, $T$                         &\small  1                                  \\ \hline
\small Signal bandwidth, $B$                            & \small 20 MHz                             \\ \hline
\small VLC AP fixed power consumption, $Q_{\rm{VLC}}$                            & \small 4 Watt                             \\ \hline
\small RF AP fixed power consumption, $Q_{\rm{RF}}$                            & \small 6.7 Watt                             \\ \hline
\small Noise PSD, $N_0$                                 & \small ${\rm{20}}^{-21}$ ${\rm{A}}^2$ /Hz \\ \hline
\end{tabular}
\end{table}
In this section, we utilize the derived analytic results in the  proposed theoretical framework alongside the corresponding  Monte Carlo simulation for validation purposes. To that end and unless otherwise stated, the analytical results in the presented figures are represented by solid lines whereas the corresponding Monte Carlo simulation results by are represented by markers.
 Our results exhibit  a perfect match between analysis and simulations in all consider cases, which justifies the validity of the derived simple closed-form expressions.
The default simulation parameters in the indicative demonstrated scenarios are depicted in Table I, unless otherwise specified.

\subsection{Sum Rate of VLC using NOMA/OFDMA}

Fig. 3 demonstrates the achieved sum rate as a function of the LED semi-angle in the VLC system with a different number of users, and different $\sigma_e$ values. It is clear from Fig. 3 that the CSI estimation error   impacts considerably the overall performance of the system. Specifically, it is shown that the maximum sum rate is decreased from 3.19 bits per channel unit (bpcu) to 1.36 bpcu a 57\% decrease ratio for the $K=2$ case. As for the considered cases of 5 and 10 users, almost the same reduction ratio applies as the CSI estimation error is increased.
Moreover, it is  noticed that high estimation errors lead to a diminishing difference between the sum rate of 2, 5 and 10 users. In other words, as the CSI estimation error increases, the channels between the users become less distinctive, which leads to minimal performance gain as the number of users  increases. This is a critical finding since NOMA and  VLC scenarios relies crucially on the state of the involved wireless channels. 
\begin{figure}
  \centering
\includegraphics[height=6.5cm, width=9cm]{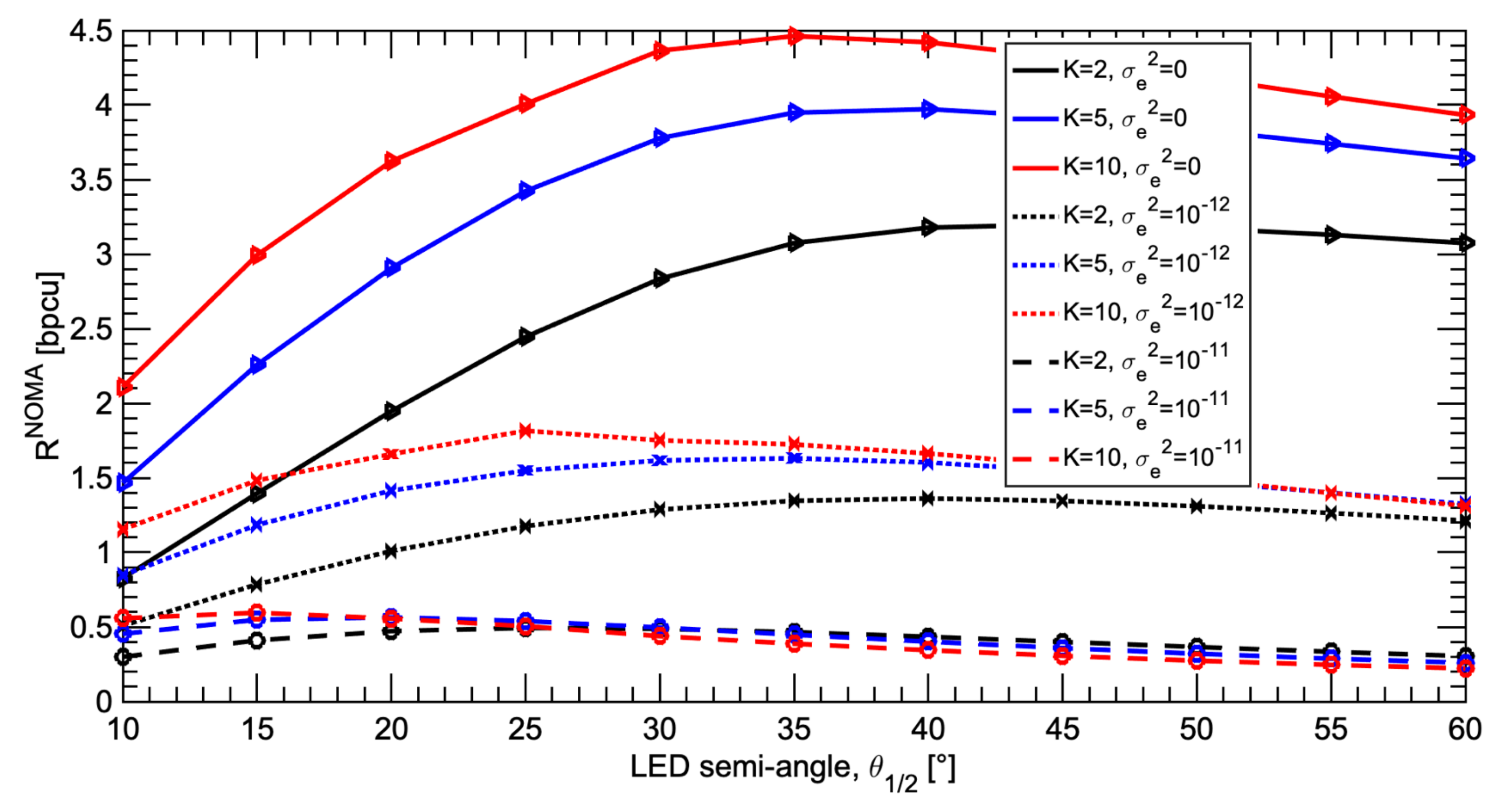}
\caption{Sum rate vs. semi-angle for LED in NOMA-VLC system.}
\label{fig2}
\end{figure}
 In the  same context, Fig. 4 illustrates the corresponding sum rate  as a function of the transmit SNR $\rho$ for  $K=2$ and $K=5$   and different $\sigma_e$ values. The first key observation is that the impact of the CSI error is rather minimal in the region of low SNR, and it becomes more pronounced as the transmit SNR $\rho$ increases. Secondly, there is an upper ceiling bound for the sum rate, leading to a flat sum-rate value.
Next, the cost of CSI errors is increased as the transmit SNR increases. For instance, at 180 dB for $K=5$, the CSI error of $-120$ dB decreases the sum rate from 12 bpcu to 3.9 bpcu, whereas in the 150 dB case, the sum rate is reduced from 7.5 bpcu to 3.9 bpcu.
Finally, in the case of CSI $-110$ dB, the performance is almost flat with a very low sum rate (around 0.2 bpcu). This verifies the importance of taking the incurred CSI imperfections into account during the  design process in order to achieve reliable and robust operation of NOMA based VLC  systems. 

\indent Fig. 5 shows the sum-rate versus $\rho$ for both NOMA-VLC and OFDMA-VLC set ups  with $K=5$ users  under different CSI error values.
As expected, NOMA outperforms OFDMA throughout the entire  SNR range considered. Moreover, the performance gain of NOMA over OFDMA is superior under the perfect CSI scenario. For example, in the case of $\rho=180$ dB, the sum rate of NOMA is almost double the sum rate of OFDMA (12.1 for the former vs. 6.4 for the latter). Interestingly, for the SNR region between 110 dB and 130 dB, NOMA with imperfect CSI and $-120$ dB outperforms OFDMA with perfect CSI. This is in fact understandable due to the performance gain for $K=5$.
Similar to NOMA, there is a limiting upper bound for OFDMA which is reached at lower SNR values (e.g. 110 dB) as compared to NOMA (e.g. 130 dB).

Fig. 6 demonstrates  the impact of $K$ on the sum rate, under different CSI error values.
As expected for NOMA-VLC with perfect CSI, the number of users is directly proportional to the sum rate, which increases from 6.4 bpcu to 7.8 bpcu, when the number of users increased from 2 to 7. However, for the case of imperfect CSI, the performance gain becomes almost flat even at an  increase of  the number of users.
This is a rather important observation since the imperfect CSI is highly disadvantageous for NOMA and a clear bottleneck. 
Yet, for OFDMA, even though it appears to be  almost flat at a change from $K=2$ to $K=7$, it can be noticed that the sum-rate decreases from 3 bpcu in the perfect CSI case to less than 0.72 bpcu and 0.59 bpcu for the cases of $-120$ dB and $-130$ dB, respectively.
Likewise, Fig. 7 shows the impact of vertical distance from the LED on the NOMA-VLC sum rate for different CSI error values.
In the case of perfect CSI, there is an apparent reduction in the sum rate from 3.76 bpcu at $L=2$ to 2.8 bpcu at $L=4.25$m. 
However, in the case of imperfect CSI  with $\sigma_e=10^{-12}$, the reduction ratio still holds from 1.5 bpcu to 0.9 bpcu.
Finally, in the case of $\sigma_e=10^{-12}$, the performance impact of $L$ becomes smaller as it varies by 0.2 m.

Fig. 8 illustrates  the sum rate as a function of $\sigma_e$. First, we can quantify  how susceptible  the VLC system is to the CSI error, since increasing it from $-120$ dB to $-110$ dB  results to a sum rate reduction from 1.6 bpcu to 0.35 bpcu for $K=5$, and from 1.4 bpcu to 0.4 bpcu for  $K = 2$. Therefore, the impact of the CSI error parameter is substantial, especially when $K=5$.
Finally, Fig. 9 shows a comparison between the OFDMA-VLC and the NOMA-VLC sum-rate performance, with respect to the corresponding  CSI error. 
It is evident  that both cases are highly susceptible to incurred CSI errors.
 For NOMA-VLC with $K=2$, the sum rate drops from  11.4 bpcu to 0.06 bpcu, as the CSI error increases from $-220$  dB to $-110$ dB. Likewise, for OFDMA based VLC, there is a dramatic decrease, which is from 10.6 bpcu to 0.07 bpcu throughout the same CSI error range.
Yet, the advantage of NOMA prevails when increasing the number of users from 2 to 5, in which case  the achieved NOMA sum rate increases from 11.4 bpcu to 12.3 bpcu, whereas OFDMA based sum rate suffers from performance degradation as it decreasing from 10.6 bpcu to 4.6 bpcu due to the orthogonality.
\begin{figure}
\centering
\includegraphics[height=7cm,width=9cm]{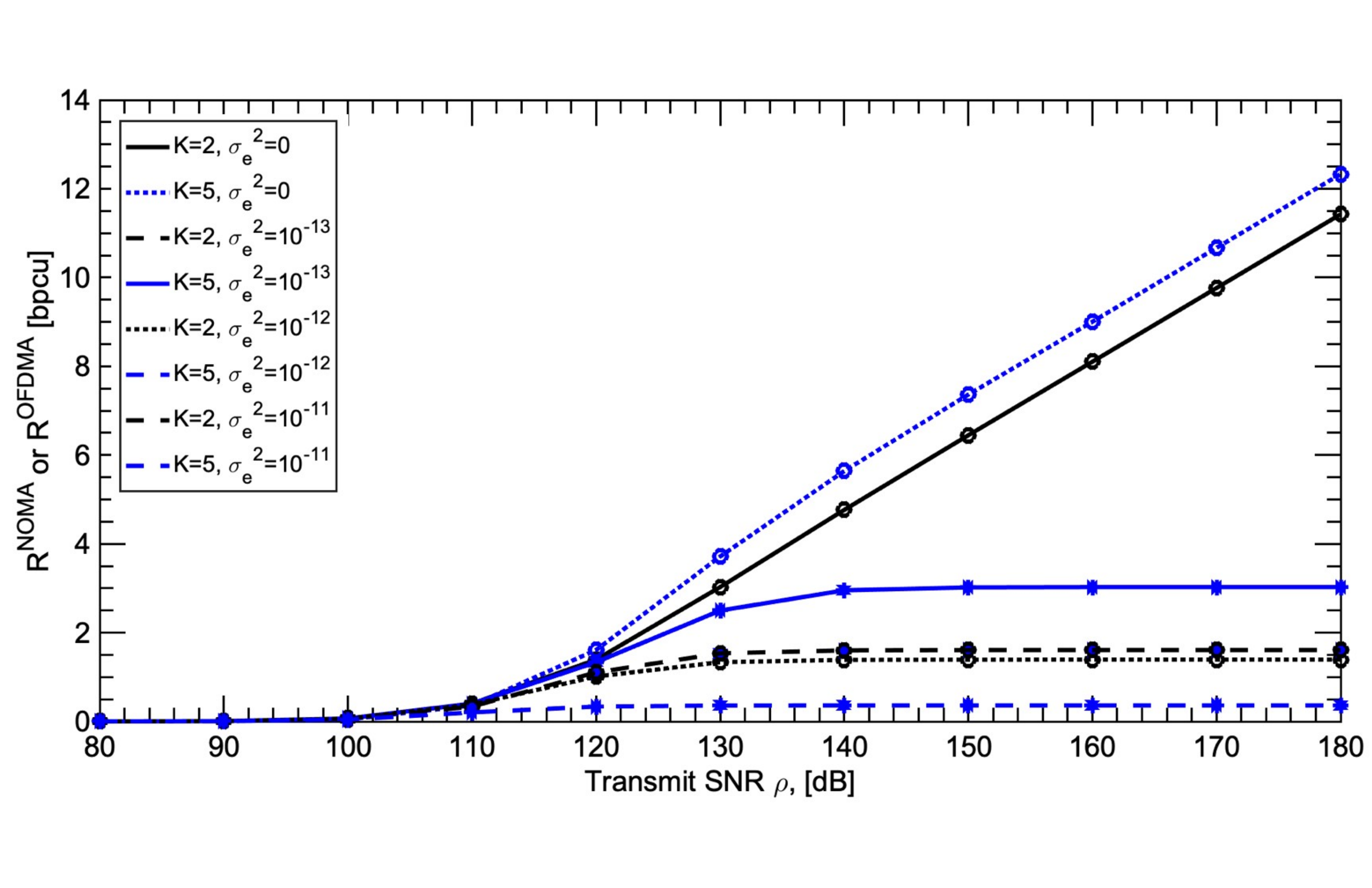}
  \caption{Sum rate vs transmit SNR in NOMA-VLC system.}
\label{fig2}
\end{figure}

  Fig. 10 illustrates  the average energy efficiency as a function of the probability of LOS availability for the case of $K = 4$, with each user  being capable of receiving both RF and VLC signals simultaneously. Finally, the transmit SNR is 150 dB for the VLC network and 30 dB for the RF network, while we assume equal probability of LOS component availability for both RF and VLC,  i.e. $\rho_{\rm{VLC}}$=$\rho_{\rm{RF}}$=$\rho$.
It is evident from Fig. 10 that the introduced  NOMA-VLC-RF system with imperfect CSI is highly sensitive to both moderate and high CSI errors. 
For instance, when the LOS availability probability is 0.5, the average energy efficiency drops from $1.55 \times 10^{7}$  bit/J to $0.65 \times 10^{7}$ bit/J in the case of moderate CSI error, and to $0.5 \times 10^{7}$ bit/J in the case of high CSI error. In other words, there is a 59\% to 70\% loss in terms of average energy efficiency when considering the practical case of CSI error.
Therefore, and given that the effect  of  CSI errors are usually neglected, it is of paramount importance to take these effects into thorough consideration during practical designs of conventional VLC ir hybrid VLC/RF systems.

\subsection{Energy Efficiency of Hybrid NOMA based RF/VLC}

This subsection analyzes the average energy efficiency of the considered hybrid VLC-RF network. To this end, a room of size 4m$\times$ 4m$\times$ 2.5m with a single VLC AP and RF AP installed on the ceiling is considered. The VLC AP has a fixed power of 4 Watt, and the RF AP has a fixed power of 6.7 W \cite{25}. The conversion efficiency of each VLC AP is 1 Watt/Amps, and its half-power angle is 50$^\circ$. The aim is to quantify the efficiency of the system when adopting NOMA within a hybrid VLC-RF network under imperfect CSI, which has been shown to be detrimental in these systems. 
In addition, the case of imperfect CSI and uniformly distributed users with a fixed power allocation policy for all users is   analyzed.

The second observation is that as the LOS availability parameter increases, the performance degradation between the perfect CSI scenario and the imperfect CSI scenario increases significantly. For instance, when the LOS availability probability is 0.1, the average energy efficiency drops from $1 \times 10^{7}$ bit/J to $0.63 \times 10^{7}$ bit/J in the case of moderate CSI error, and to $0.48 \times 10^{7}$ bit/J in the case of high CSI error. However, in the case of full LOS availability, the average energy efficiency drops from $2.57 \times 10^{7}$ bit/J to $0.9\times 10^{7}$ bit/J in the case of moderate CSI error, and to $0.58 \times 10^{7}$ bit/J in the case of high CSI error.
\begin{figure}
\includegraphics[height=7cm,width=9cm]{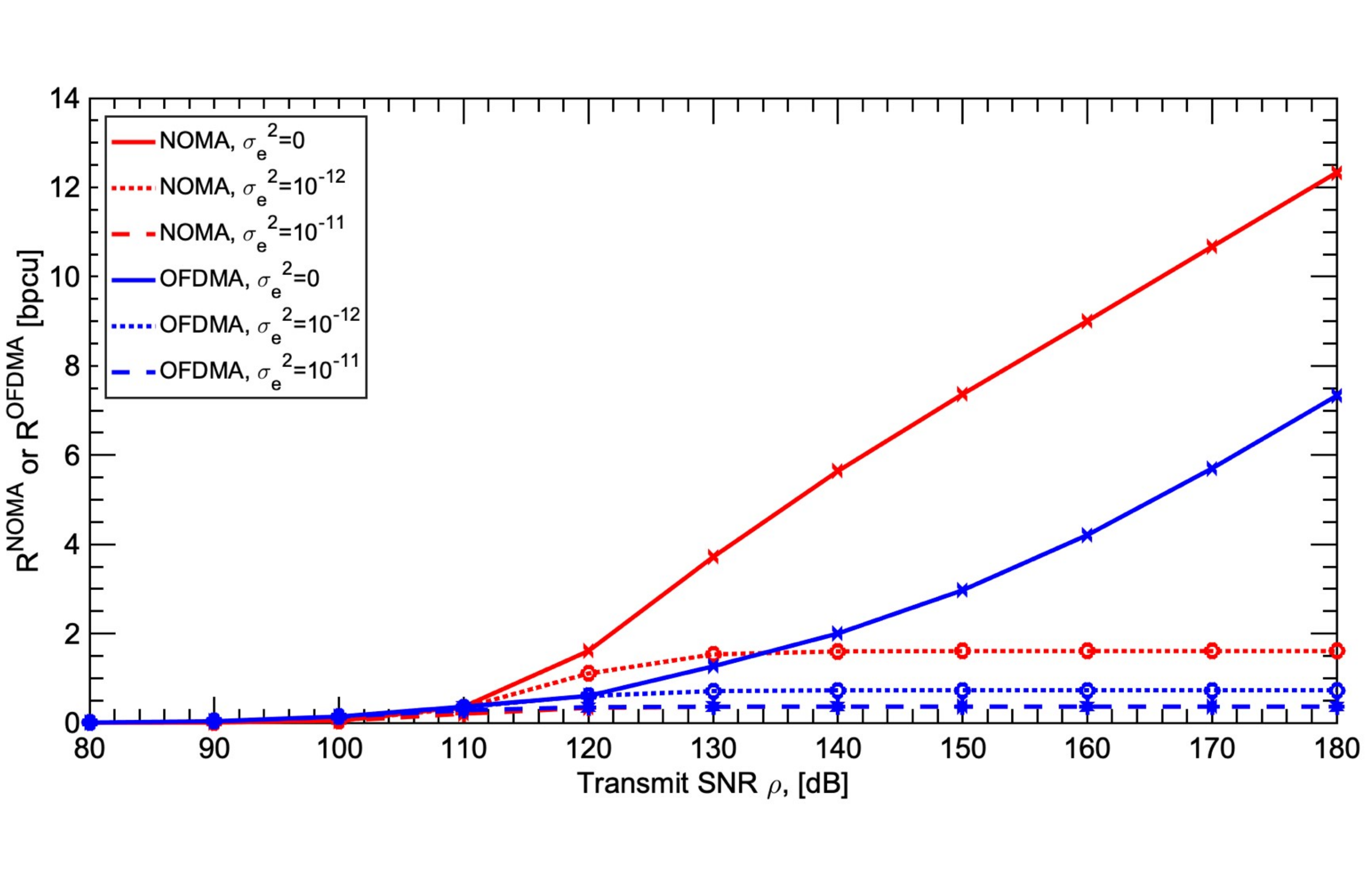}
\caption{Sum rate Vs transmit SNR for NOMA-VLC and OFDMA-VLC}
\label{fig2}
\end{figure}
\begin{figure}
\centering
\includegraphics[height=6cm,width=9cm]{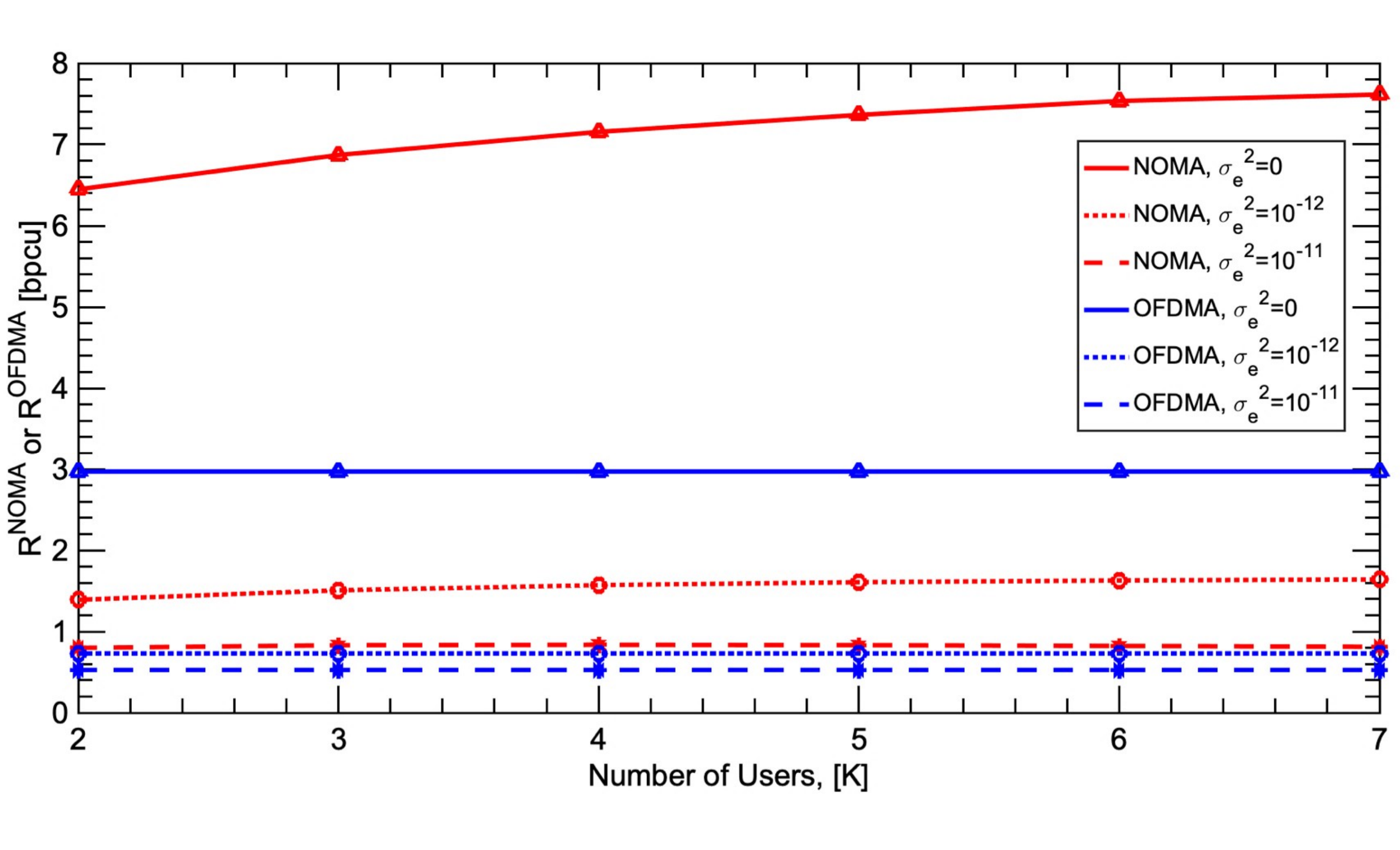}
  \caption{Sum rate vs number of users for NOMA-VLC}
\label{fig2}
\end{figure}

In the same context, the average energy efficiency is illustrated in Fig. 11 as a function of VLC fixed power consumption. Under this scenario, two LOS conditions, \it{i.e.} \rm , 0.5 and 1 at high CSI error are considered.
The average energy efficiency is severely degraded when CSI error exists as compared to the case of error-free CSI. Moreover, the average energy efficiency reduced from $3.3 \times 10^{7}$ bit/J to $2.3 \times 10^{7}$ bit/J which is almost a  30\% drop. However, in the case of imperfect CSI, the average energy efficiency is almost flat regardless of the VLC fixed power consumption value. As a result, we can conclude that as the CSI error increases, reducing the VLC fixed power consumption will not result to any noticeable performance gain in terms of bits/J. Therefore, it is advised to mitigate the root cause of the issue which is the CSI error.
 \begin{figure}
\centering
\includegraphics[height=6.9cm,width=9cm]{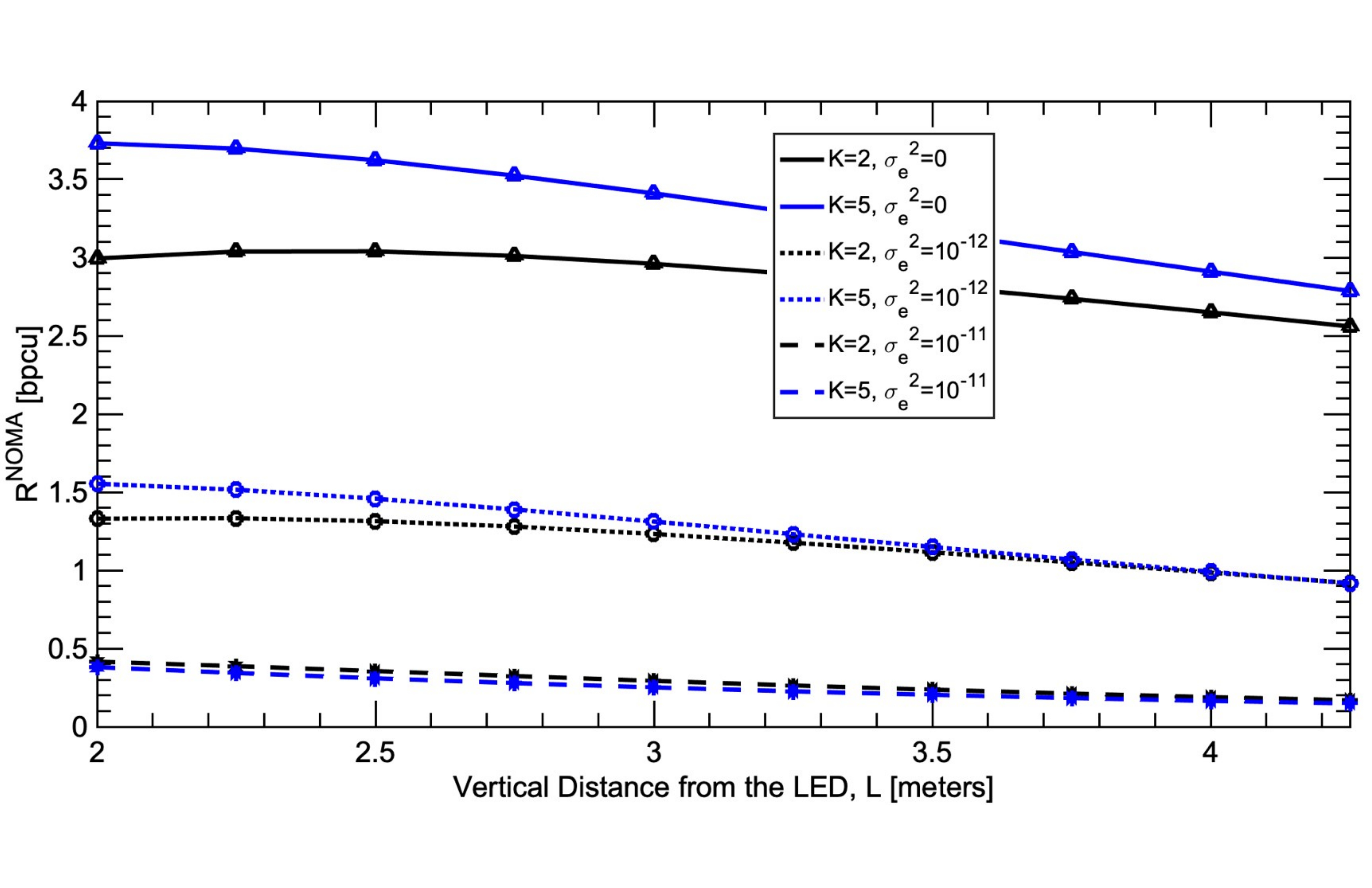}
  \caption{Sum rate vs. vertical distance from LED for NOMA-VLC}
\label{fig2}
\end{figure}
\begin{figure}
\centering
\includegraphics[height=7cm,width=9cm]{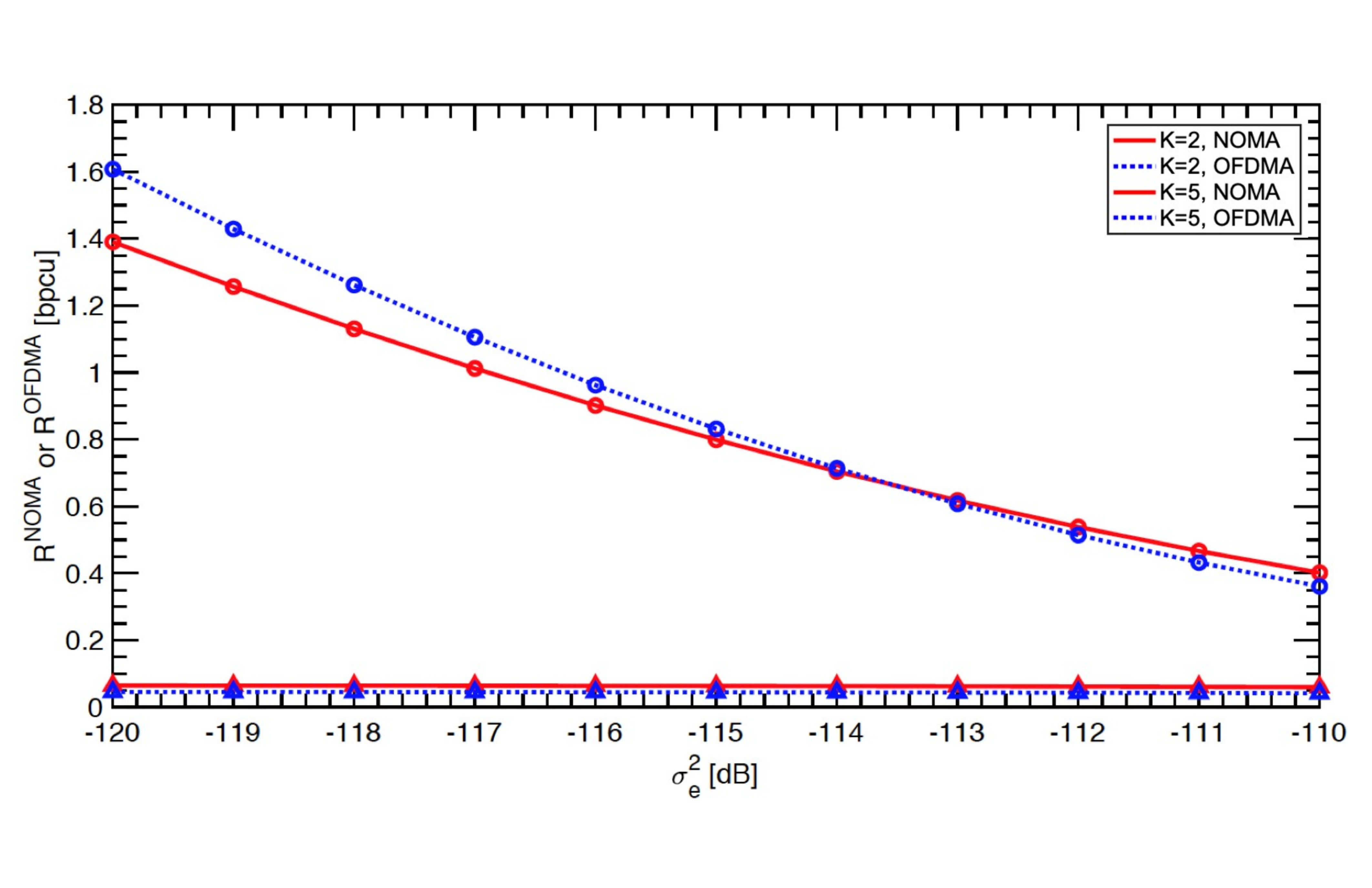}
  \caption{Sum rate Vs CSI error (NOMA-VLC)}
\label{fig2}
\end{figure}

\begin{figure}
\centering
\includegraphics[height=7cm,width=9cm]{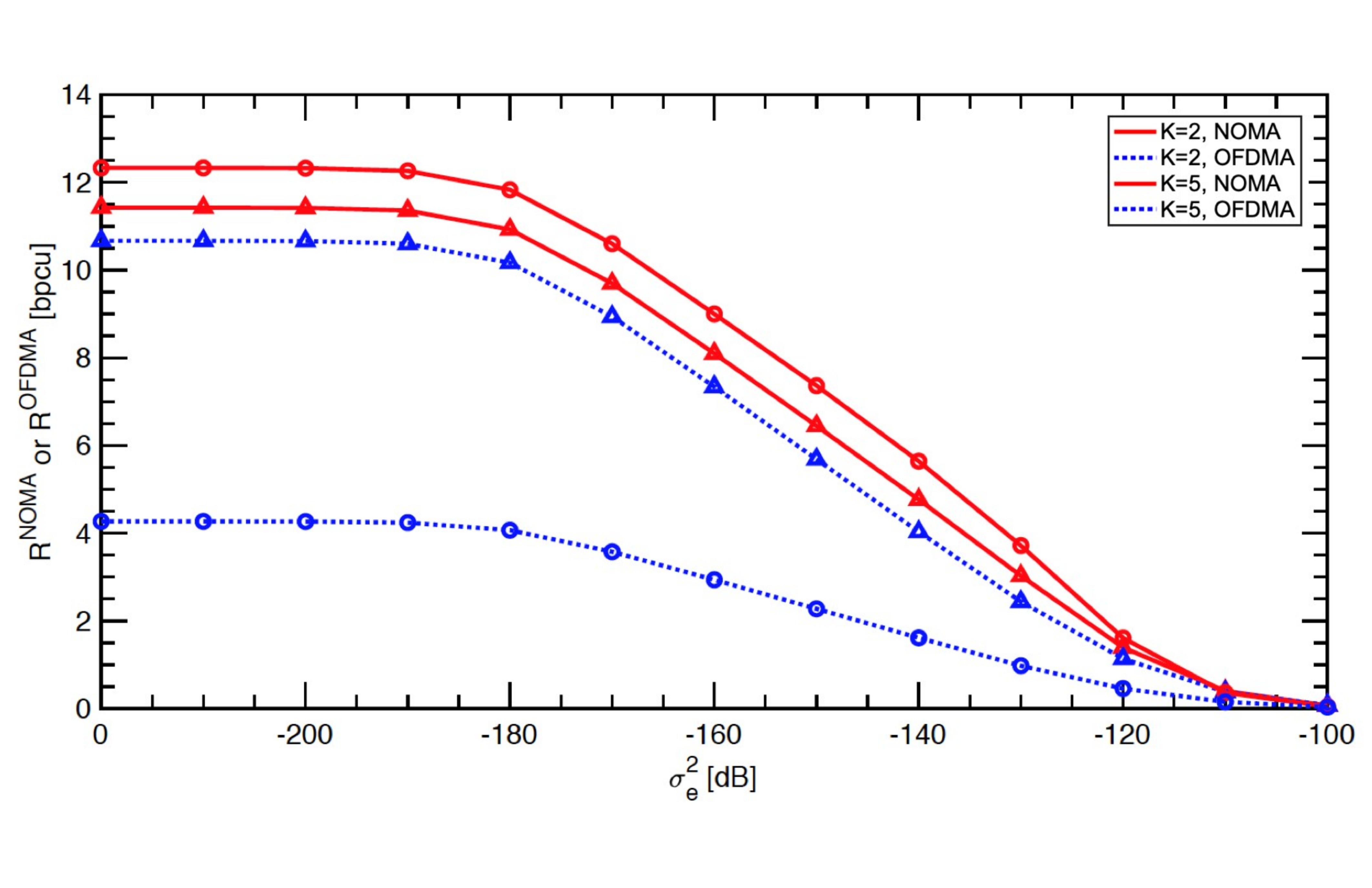}
  \caption{Sum rate Vs CSI error (NOMA-VLC and OFDMA-VLC)}
\label{fig2}
\end{figure}
\begin{figure}
\centering
\includegraphics[height=6.2cm,width=9cm]{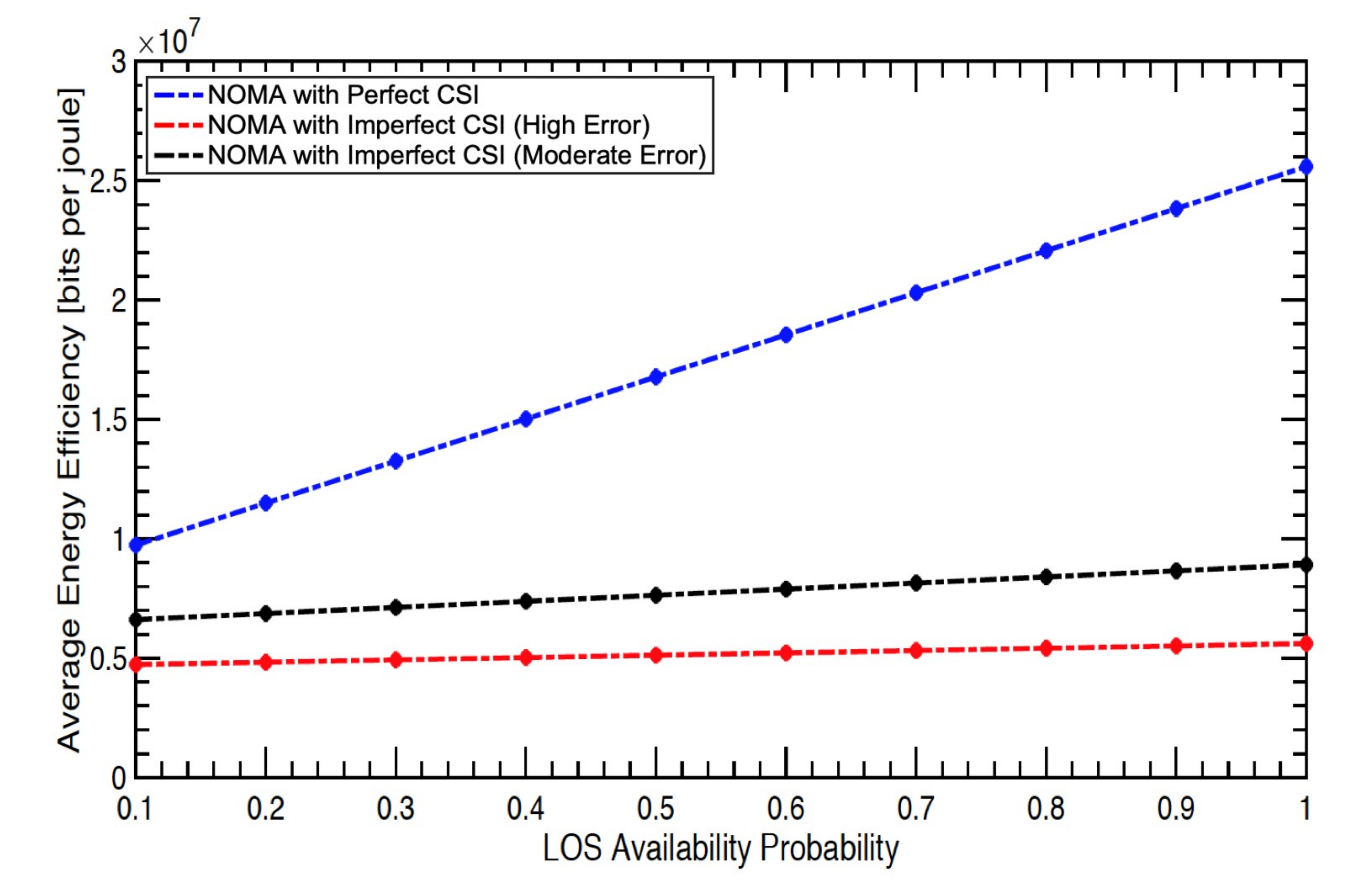}
\caption{Average energy efficiency vs. LOS availability probability with NOMA under different CSI conditions, with 4 users.}
\label{fig2}
\end{figure}
\begin{figure}
\centering
\includegraphics[height=7cm,width=9cm]{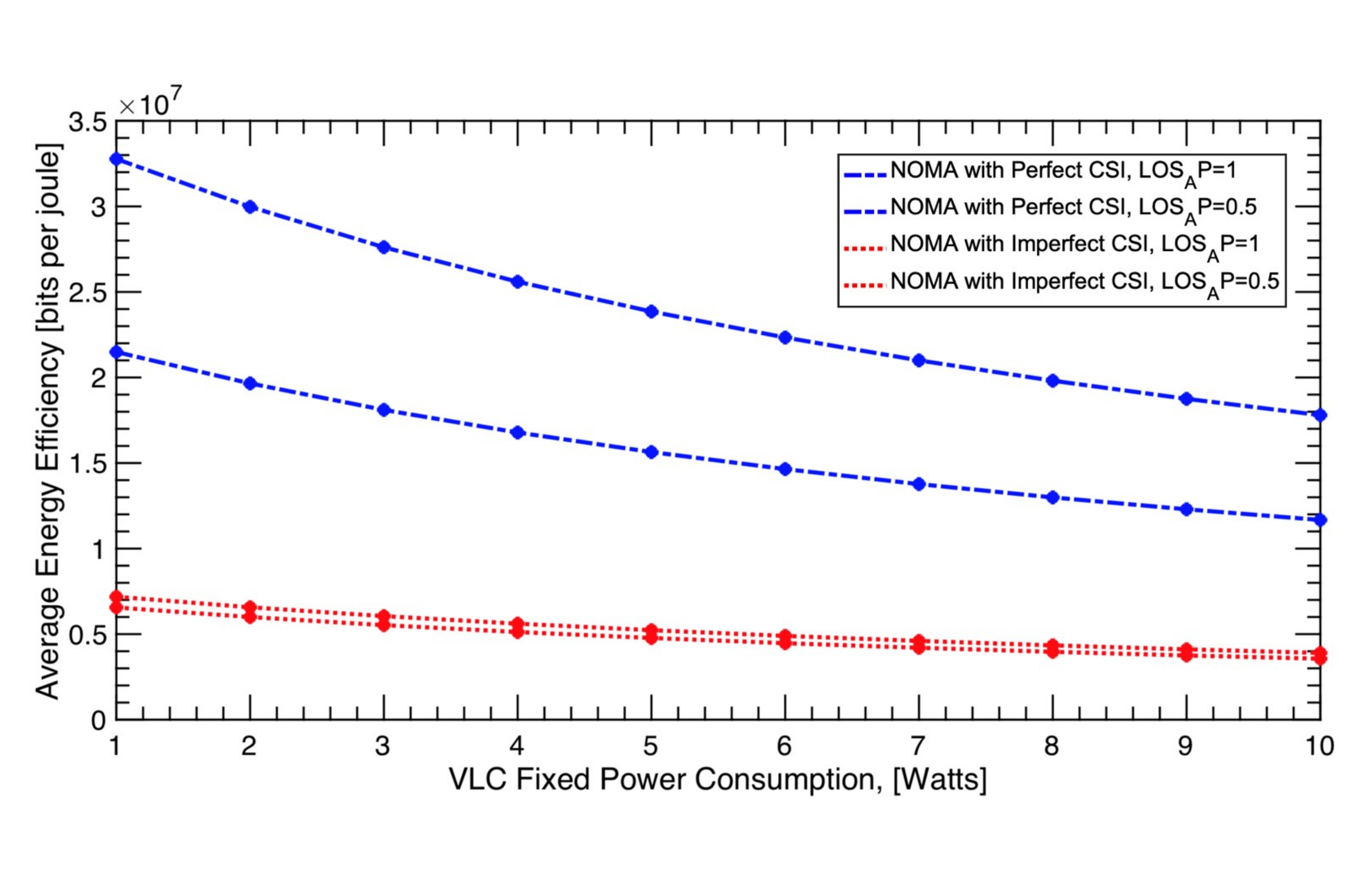}
\caption{Average energy efficiency vs. VLC AP power consumption using the NOMA scheme and the OFDMA scheme, with 4 users.}
\label{fig2}
\end{figure}

\indent Fig. 12 presents the average energy efficiency as a function of the number of users, using two LOS conditions (0.5 and 1), and with both perfect and imperfect high CSI error. The advantageous performance of NOMA becomes evident  since it exhibits an almost unaffected average energy efficiency performance while increasing the number of users. However, two main observations can be drawn from this figure.
First, the impact of LOS becomes negligible in case of high CSI error. For instance, in the case of 4 users with perfect CSI knowledge, the average energy efficiency decreases from  $2.5\times 10^{7}$ bit/J to $1.75\times 10^{7}$ bit/J. However, in the case of imperfect CSI, the performance drops from $0.6\times 10^{7}$ bit/J to $0.55\times 10^{7}$ bit/J.

\indent Finally, Fig. 13 demonstrates the average energy efficiency as a function of the LOS availability probability, under perfect and imperfect CSI in the presence of high CSI error. In the case of perfect CSI, as expected, the VLC-Only system outperforms the hybrid VLC-RF system when the probability of LOS availability is higher than 0.15. However, in the practical scenario where the CSI is imperfect, the advantages of the robust hybrid VLC-RF scheme are shown against the prone standalone VLC system. Notably, the average energy efficiency for the standalone VLC system with imperfect CSI is close to zero when there is a dominant NLOS scenario. However, for the hybrid VLC-RF system, the average energy efficiency withstands the severe LOS situation with an acceptable value of $0.77\times 10^{7}$ bit/J.
In addition, as the LOS availability probability reaches to 1 in the case of imperfect CSI, the average energy efficiency of both schemes meet at $0.9\times 10^{7}$ bit/J. This means that the hybrid VLC-RF scheme is considerably more robust to withstand all the cases of LOS and can only be better when there is a partial LOS, or can provide equal performance to the VLC-Only system when the LOS is fully available.

\begin{figure}
\centering
\includegraphics[height=7cm,width=9cm]{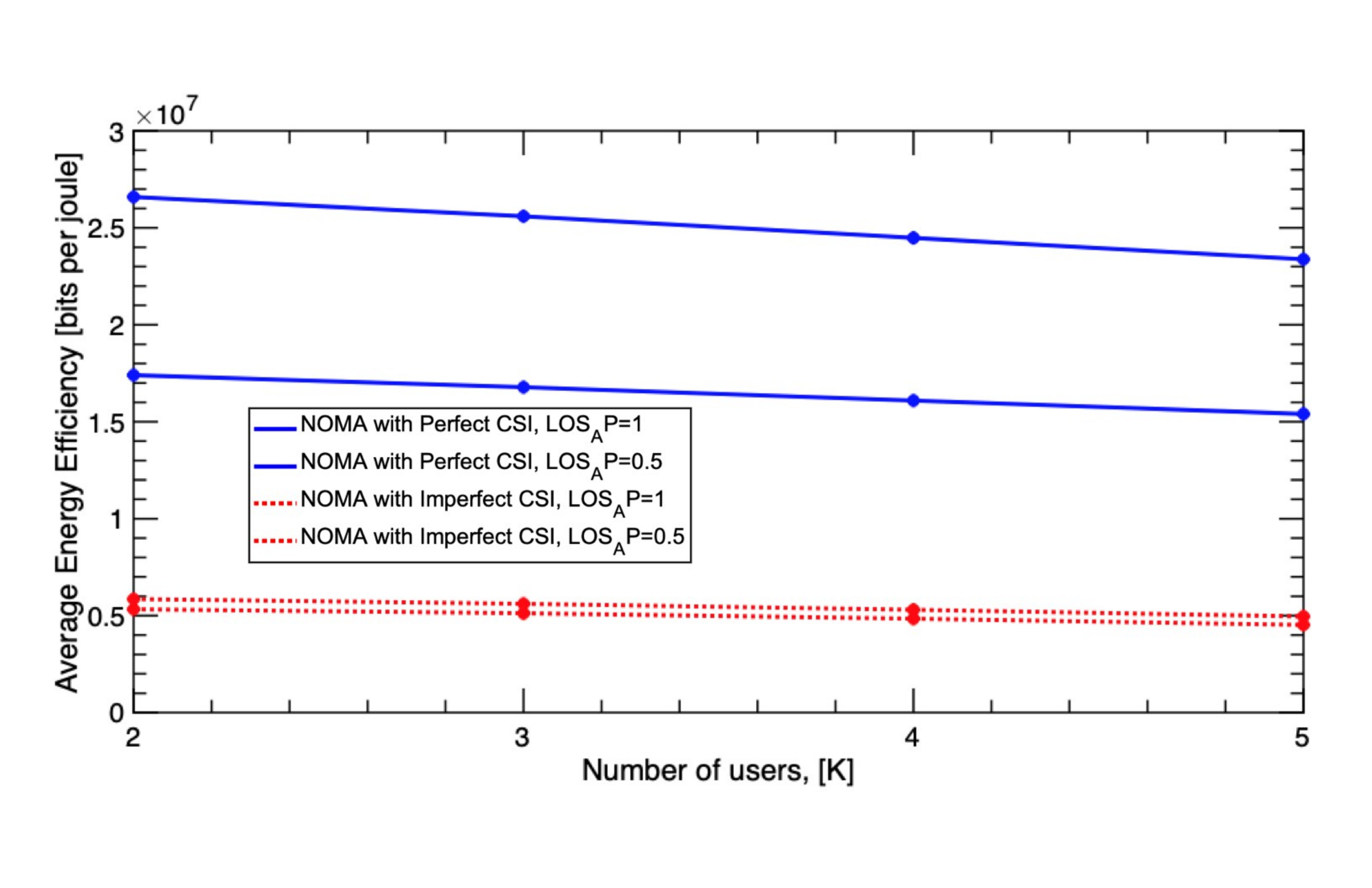}
\caption{Average energy efficiency vs. number of users, under imperfect CSI}
\label{fig2}
\end{figure}
\begin{figure}
\includegraphics[height=6.5cm,width=9.5cm]{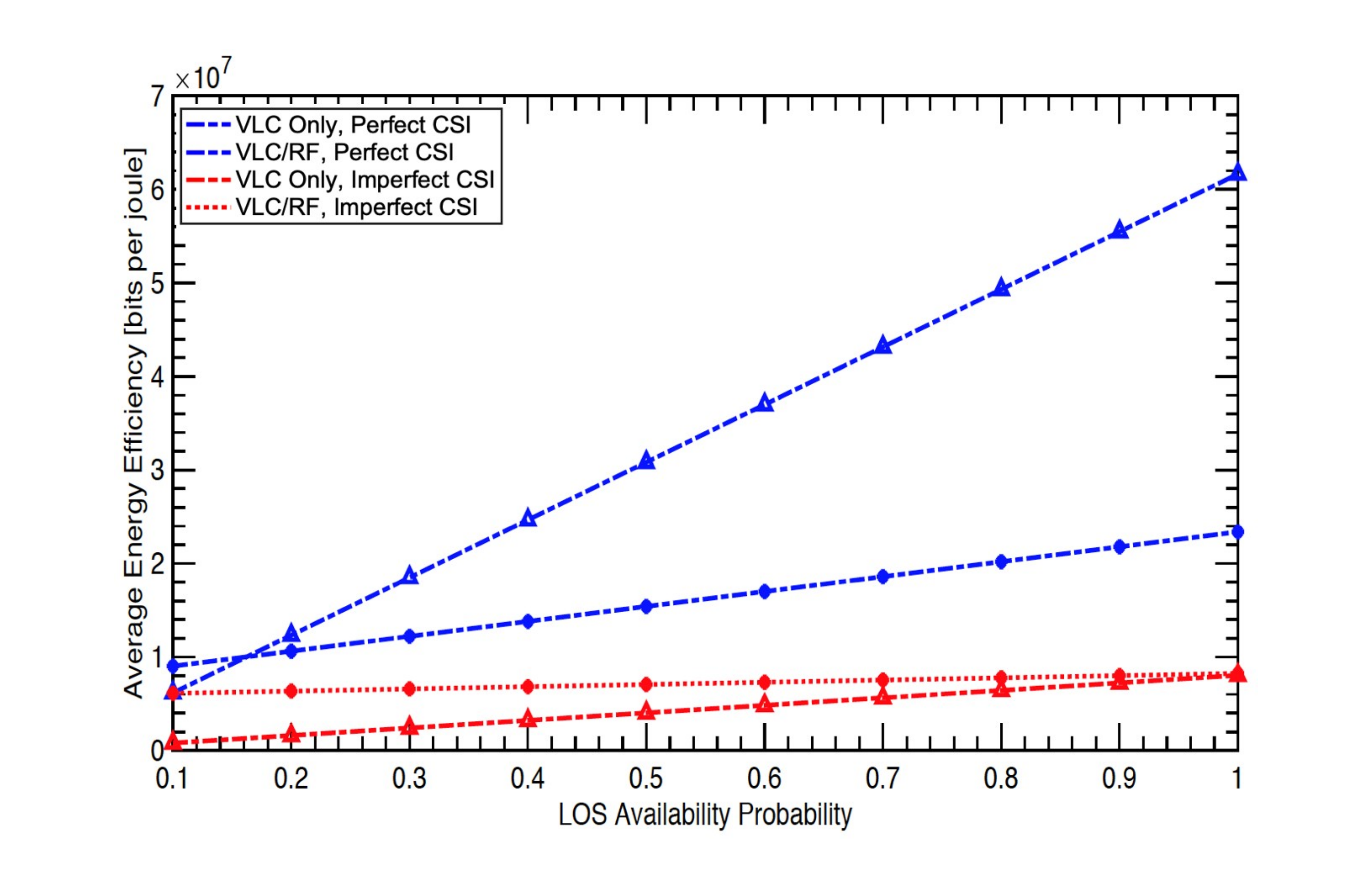}
\caption{Average energy efficiency vs. LOS availability probability, using standalone VLC Scheme, and Hybrid VLC-RF Scheme, under imperfect CSI, with 4 users.}
\label{fig2}
\end{figure}

\section{Conclusion}
This work evaluated the performance of a hybrid NOMA-VLC-RF system, assuming imperfect CSI and uniformly distributed users. Closed-form expressions were derived for the corresponding average sum rate and average energy efficiency, which were extensively verified through comparisons with respective computer simulation results. The results showed that the amount of CSI error has a considerable impact on the performance for both NOMA and OFDMA schemes.
 Although NOMA outperforms OFDMA in all the cases, it was also shown that NOMA-VLC systems with a relatively high number of users are more sensitive compared to the NOMA-VLC system with a lower number of users.
Moreover, it was shown that the CSI error poses a considerable impact on the average energy efficiency of the hybrid NOMA-VLC-RF system. It was shown that decreasing the CSI error yields a much higher increase in the average energy efficiency, as compared to merely decreasing the VLC power consumption.
The work also demonstrated that the hybrid scheme of VLC-RF with NOMA is much more robust than the VLC-Only System with NOMA, under the case of imperfect CSI. It was also shown that the impact of CSI errors is the most dominant factor affecting the average energy efficiency when compared to other parameters such as  the number of users and VLC power consumption.
Finally, the reported results can highly aid the design of practical hybrid VLC-RF systems, since they reveal the high impact of CSI error, whereby a small difference in the amount of error may lead to a considerable impact on the average sum rate and the average energy efficiency of realistic communication scenarios.

 \begin{center}
\bibliographystyle{IEEEtran}

\end{center}

\end{document}